\documentclass[journal,twoside,web]{ieeecolor}
\usepackage{generic}
\usepackage{cite}
\usepackage{amsmath,amssymb,amsfonts}
\usepackage{algorithmic}
\usepackage{graphicx}
\usepackage{textcomp}
\usepackage{xcolor}
\usepackage[scr]{rsfso}
\usepackage{soul}
\usepackage{hyperref}

\newtheorem{definition}{Definition}
\newtheorem{theorem}{Theorem}
\newtheorem{remark}{Remark}
\newtheorem{example}{Example}
\newtheorem{corollary}{Corollary}
\newtheorem{lemma}{Lemma}
\newtheorem{proposition}{Proposition}

\def\BibTeX{{\rm B\kern-.05em{\sc i\kern-.025em b}\kern-.08em
    T\kern-.1667em\lower.7ex\hbox{E}\kern-.125emX}}
\begin{document}
\title{Neural Operators for 
Bypassing Gain and Control Computations in PDE Backstepping}
\author{  Luke Bhan,
        Yuanyuan Shi, 
        and Miroslav Krstic
\vspace*{-3em}
        \thanks{The authors are with the University of California, San Diego, USA, lbhan@ucsd.edu, yyshi@eng.ucsd.edu, {krstic@ucsd.edu}}
}

\maketitle


\begin{abstract}
We introduce a framework for eliminating the computation of controller gain functions in PDE control. We learn the nonlinear operator from the plant parameters to the control gains with a (deep) neural network. We provide closed-loop stability guarantees (global exponential) under an NN-approximation of the feedback gains. While, in the existing PDE backstepping, finding the gain kernel requires (one offline) solution to an integral equation, the neural operator (NO) approach we propose learns the mapping from the functional coefficients of the plant PDE to the kernel function by employing a sufficiently high number of offline numerical solutions to the kernel integral equation, for a large enough number of the PDE model's different functional coefficients. We prove the existence of a DeepONet approximation, with arbitrarily high accuracy, of the exact nonlinear continuous operator 
mapping PDE coefficient functions into gain functions. 
Once proven to exist, learning of the NO is standard, completed ``once and for all'' (never online) 
and the kernel integral equation doesn't need to be solved ever again, for any new functional coefficient 
not exceeding the magnitude of the functional coefficients used for training. 
We also present an extension from approximating the gain kernel operator to approximating the full feedback law mapping, from plant parameter functions and state measurement functions to the control input, with semiglobal practical stability guarantees.
Simulation illustrations are provided and code is available on \href{https://github.com/lukebhan/NeuralOperatorsForGainKernels}{github}. This framework, eliminating real-time recomputation of gains, has the potential to be game changing for adaptive control of PDEs and gain scheduling control of nonlinear PDEs.

The paper requires no prior background in machine learning or neural networks.
\end{abstract}


\vspace*{-1em}

\section{Introduction}

ML/AI is often (not entirely unjustifiably) thought of as an `existential threat' to model-based sciences, from physics to conventional control theory. In recent years, a framework has emerged \cite{lu2019deeponet,lu2021deeponet, li2020neural,li2021fourier}, initiated by George Karniadakis, his coauthors, and teams led by Anima Anandkumar and Andrew Stuart, which promises to unite the goals of physics and learning, rather than presenting learning as an alternative or substitute to first-principles physics. In this framework, often referred to as neural operators (NO), which is formulated as learning of mappings from function spaces into function spaces, and is particularly suitable for PDEs, solution/``flow'' maps can be learned after a  sufficiently large number of simulations for different initial conditions. (In some cases,  parameters of models can also be identified from experiments.)


\paragraph{Mappings of plant parameters to control gains and learning of those maps} 
One can't but ask
what the neural operator reasoning can offer to control theory, namely, to the design of controllers, observers, and online parameter estimators. This paper is the first venture in this direction, a breakthrough with further possibilities, and a blueprint (of a long series of steps) to learn PDE control designs and prove their stability. 

In control systems (feedback controllers, observers, identifiers), various kinds of nonlinear maps arise, some from vector  into vector spaces, others from vector or function spaces into function spaces. Some of the maps have time as an argument (making the domain infinite) and others are mappings from compact domains into compact image sets, such as mappings converting system coefficients into controller coefficients, such as the mapping $K(A,B)$ for the closed-loop system $\dot x = Ax + Bu,\ u=Kx$ (under either pole placement or LQR). 

While learning nonlinear maps for various design problems for nonlinear ODEs would be worth a study, we focus in this initial work one step beyond, on a benchmark PDE control class. Our focus on an uncomplicated---but unstable---PDE control class is for pedagogical reasons. Combining the operator learning with
{\em PDE backstepping} is complex enough even for the simplest-looking among PDE stabilization problems. 

\paragraph{PDE backstepping control with the gain computation obviated using neural operators}
Consider 1D hyperbolic partial integro-differential equation systems of the general form $v_t(x,t) = v_x(x,t) + \lambda(x) v(x,t) + g(x) v(0,t) +\int_0^x f(x,y) v(y,t) dy$ on the  unit interval $x\in[0,1]$, which are  transformable, using an invertible backstepping ``pre-transformation'' introduced in \cite{Bernard2014} into the simple PDE 
\begin{eqnarray}
    \label{eq-PDE}
	u_t(x,t) &=& u_x(x,t) + \beta(x) u(0, t) 
 \\  \label{eq-PDEBC}
	u(1, t) &=& U(t). 
\end{eqnarray}
Our goal is the design of a PDE backstepping boundary control 
\begin{equation}\label{eq-bkstfbkk}
    U(t) = \int_0^1 k(1-y) u(y,t) dy. 
\end{equation}
Physically, \eqref{eq-PDE} is a ``transport process (from $x=1$ towards $x=0$) with recirculation'' of the outlet variable $u(0,t)$. Recirculation causes instability when the  coefficient $\beta(x)$ is positive and large. This instability is prevented by the backstepping boundary feedback \eqref{eq-bkstfbkk} with the gain  function $k(\cdot)$ as a kernel in the spatial integration of the measured state $u(y,t)$. (The full state does not need to be measured, as explained in Remark \ref{rem-observer} at the end of Section \ref{sec-stabilityDeepONet}.)

Backstepping  produces the gain kernel $k$ for a given  $\beta$. The mapping ${\cal K} : \beta \mapsto k$ is nonlinear, continuous, and we learn it. 

Why do we care to learn ${\cal K}$? The kernel {\em function} $k$ can always be computed for a particular $\beta$, so what is the interest in learning the {\em functional mapping/operator}? Once ${\cal K}$ is learned, $k$ no longer needs to be sought, for a new $\beta$, as a solution to a partial differential or integral equation. For the next/new $\beta$, finding $k$ is simply a ``function evaluation'' of the learned mapping ${\cal K}$. This provides  benefits in 
both adaptive control where, at each time step, the gain  estimate $\hat k$ has to be computed for a new parameter update $\hat\beta$, and in gain scheduling for nonlinear PDEs where the gain  has to be recomputed at each current value of the state. 

\begin{figure}[t]
    \centering \includegraphics{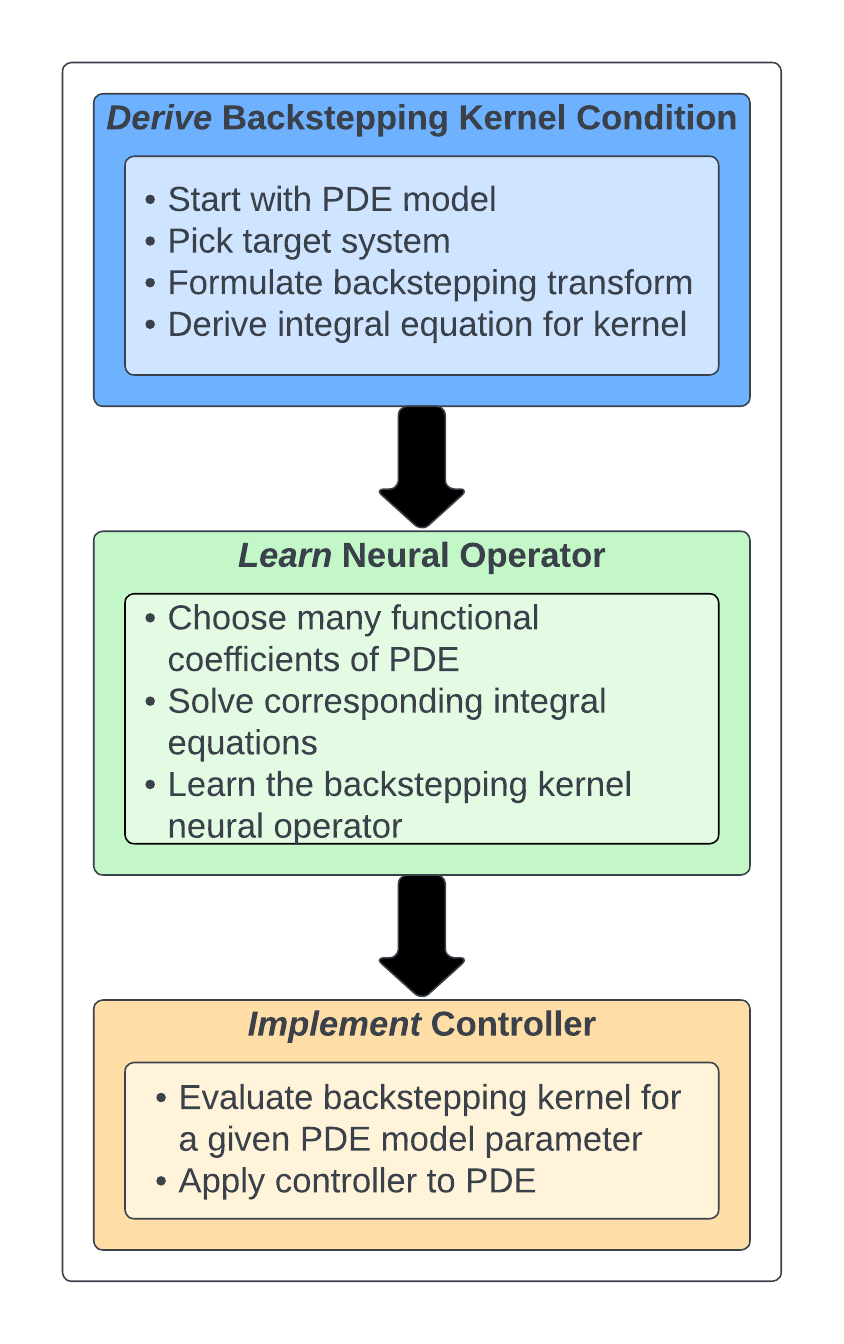}
	\caption{An algorithmic representation of our design paradigm of employing neural operators in boundary control of PDEs. Three major step clusters are performed: (1)  \underline{derivation} of the integral equations for the backstepping kernels, performed only once; (2)  \underline{learning} of the mapping from the plant parameter functions into the backstepping kernel functions, also performed only once; and (3) \underline{implementation} of the controller for specific plant parameters. The task in the top box has  been completed in \cite{krstic2008Backstepping}. In this paper, the task in the middle box is introduced and stability guarantees for the task in the bottom box are provided.}
    \label{fig:0}
\end{figure}

As well known, learning (ML, in general, and its operator learning varieties: DeepONet, FNO, LOCA, NOMAD, etc.) comes with an upfront price. Large data sets need to be first produced, and then large (possibly ``deep'') neural networks need to be trained. There is no exception to this in the approach we propose. For a large sample set of recirculation functions $\beta_i$, we need to first solve for the corresponding backstepping kernels $k_i$. After that, a NN approximation of ${\cal K}$ needs to be trained on that data set of the $(\beta_i, k_i)$ pairs. 

One can stop at producing the NN approximation of the mapping ${\cal K}$ and proceed with a heuristic use of the  approximated gains $\hat k$. But we don't stop there. We ask whether the PDE system will be still stabilized with the NN-approximated gain kernel $\hat k$. Our main theoretical result is affirmative. With a large enough data set of solved pairs $(\beta_i, k_i)$, and a large enough trained (deep) NN, closed-loop stability is guaranteed for a new $\beta$, not in the training set. 

When ML is applied in the control  context (as RL or other approaches), it is usually regarded as a model-free design. Our design, summarized in Figure \ref{fig:0}, is not model-free; it is model-based. It is only that the computational portion of this model-based (PDE backstepping) design is obviated through ML. 

Our learning is offline; not as in adaptive control \cite{Bernard2014,Anfinsen2019Adaptive}. 

\paragraph{Neural operator literature---a brief summary}
    Neural operators are NN-parameterized maps for learning relationships between function spaces. They originally gained popularity due to their success in mapping PDE solutions while remaining discretization-invariant. Generally, nonlinear operators consist of three components: an encoder, an approximator, and a reconstructor \cite{lanthaler2022errorDeeponet}. The encoder is an interpolation from an infinite-dimensional function space to a finite-dimensional vector representation. 
    The approximator aims to mimic the infinite map using a finite-dimensional representation of both the domain function space and the target function space. 
    The reconstructor then transforms the approximation output into the infinite-dimensional target function space. 
    The implementation of both the approximator and the reconstructor is generally coupled and can take many forms. For example, the original DeepONet \cite{lu2021deeponet} contains a ``branch'' net that represents the approximation network and a ``trunk'' net that builds a basis for the target function space. The outputs of the two networks are then taken in linear combination with each other to form the operator. FNO \cite{li2021fourier} utilizes the approximation network in a Fourier domain where the reconstruction is done on a basis of the trigonometric polynomials. LOCA \cite{kissas2022loca} integrates the approximation network and reconstruction step with a unified attention mechanism. 
    NOMAD \cite{seidman2022nomad} extends the linear reconstructor map in DeepONet to a nonlinear map that is capable of learning on nonlinear submanifolds in function spaces.
    There have been many more extensions to the neural operator architectures omitted here as they are usually designed around domain-specific enhancements \cite{wang2021physicsinformedNeuralOperators} \cite{li2022dissipative} \cite{Pickering2022}. Another line of work, called physics-informed neural networks (PINNs) \cite{raissi2019physics, karniadakis2021physics}, which can be used as generic solvers of PDEs by adding physics constraint loss to neural networks. However, PINNs need to be re-trained for new recirculation function $\beta$, thus not providing as much acceleration for the computation of the backstepping kernels as the neural operators.

\paragraph{Advances in learning-based control}
Among the first in demonstrating the stability of learning-based model predictive controllers (MPC) were the papers \cite{ASWANI20131216,rosolia2017learning}, followed in several directions. First, for nonlinear systems, deep learning-based approaches consist of jointly learning the controller and(or) Lyapunov functions via NNs~\cite{tedrakeNNControl, chang2019neural, chen2020learning, chen2021learning, chen2021learning2, dawson2022safe, chen2022large}. \cite{chang2019neural} proposed a method for learning control policies and NN Lyapunov functions using an empirical Lyapunov loss and then validating using formal verification. \cite{chen2020learning, chen2021learning} generalize the method to learning Lyapunov functions for piecewise linear and hybrid systems, and \cite{chen2021learning2} for learning regions of attraction of  nonlinear systems. 
In addition, \cite{taylor2019control, nguyen2021}
have explored how learning-based control will affect nominal systems with known Lyapunov functions, and 
\cite{boffi2021learning, pfrommer2022tasil, claudioNonlinear2023} studied the problem of learning stability certificates and stable controllers directly from data.
In a similar vein, \cite{frank2022} has developed a provable stable data-driven algorithm based on system measurements and prior knowledge for linear time-invariant systems. 

 In a separate, but related direction, many reinforcement learning (RL) \cite{Bert05,sutton2018reinforcement} control approaches have been developed over the past few years. On the one side, model-based RL has been studied due to its superior sample efficiency and interpretable guarantees. The main focus has been on learning the system dynamics and 
 providing closed-loop guarantees in \emph{finite-time} for both linear systems \cite{dean2018regret,chen2021black,lale2022reinforcement,faradonbeh2018finite,tsiamis2022learning} (and references within), and nonlinear systems~\cite{berkenkamp2017safe,kakade2020information,singh2021learning,lale2022kcrl}. 
 For model-free RL methods, \cite{fazel2018global,mohammadi2021convergence,jiang2022,zhao2023global} proved the convergence of policy optimization, a popular model-free RL method, to the optimal controller for linear time-invariant systems, ~\cite{PANG2020109035, pramod2022} for linear time-varying systems, ~\cite{tang2021analysis} for partially observed linear systems. See \cite{hu2022towards} for a recent review of policy optimization methods for continuous control problems such as the LQR, $H_{\infty}$ control, risk-sensitive control, LQG, and output feedback synthesis. For nonlinear systems, \cite{chow2018lyapunov,choi2020reinforcement, cui2022structured,shi2022stability} investigated policy optimization with stability guarantees in which the stability constraints are derived from control Lyapunov functions.
 In addition to policy optimization methods, \cite{vamvoudakis2010online, lewis2012reinforcement, bhasin2013novel, vamvoudakis2017q} have studied and proved the stability and asymptotic convergence of other model-free RL algorithms such as actor-critic methods~\cite{vamvoudakis2010online, lewis2012reinforcement} and Q-learning~\cite{vamvoudakis2017q} in control affine systems.
In the domain of cyber-physical systems (CPS), a theoretical framework has been developed for learning-based control to handle partially observable systems \cite{andreasSeparation2023}.

Many advances have been made in learning-based control in games and multi-agent systems \cite{zhang2019policy, mazumdar2020gradient, fiez2020implicit, mojica2022stackelberg, zhang2021multi, mao2022provably, 
poveda2022fixed, fiez2020implicit, vamvoudakis2017game, qu2020scalable, lin2021multi}. Convergence is characterized for various learning-based methods to Nash equilibria in zero-sum linear quadratic games~\cite{zhang2019policy}, continuous games~\cite{mazumdar2020gradient}, Stackelberg games~\cite{fiez2020implicit, mojica2022stackelberg}, Markov games~\cite{zhang2020model,mao2022provably}, and multi-agent learning over networked systems~\cite{qu2020scalable, lin2021multi, poveda2022fixed}. A recent review for learning-based control in games is in~\cite{zhang2021multi}. 

We focus on learning-based control for PDE systems. In our previous work~\cite{shi2022machine}, we demonstrate the empirical success of using NOs for accelerating PDE backstepping observers, without theoretical guarantees. This work represents the first step towards using NOs for provably bypassing gain computations (with exponential stability guarantees) or directly learning the controller (with practical stability) in PDE backstepping. 
\paragraph{Backstepping control of first-order hyperbolic PDEs}
The PDE system \eqref{eq-PDE}, \eqref{eq-PDEBC} is the simplest open-loop unstable PDE of any kind which can be of interest to the researcher working on PDE stabilization by boundary control. This system is treated here as a technical benchmark, as was done as well in \cite{Bernard2014} and a number of other references offering methodological advances in PDE stabilization. System \eqref{eq-PDE}, \eqref{eq-PDEBC} is a particular case of a single-PDE hyperbolic class in \cite{krstic2008Backstepping} for which PDE backstepping was first introduced in the hyperbolic setting. Coupled systems of first-order hyperbolic PDEs are of greater interest because they arise in fluid flows, traffic flows, elastic structures, and other applications. The first result on backstepping for a {\em pair} of coupled hyperbolic PDEs was in \cite{Coron2013Local}. The extension from two to $n+1$ hyperbolic PDEs, with actuation of only one and with counterconvection of $n$ other PDEs was introduced in \cite{Meglio2013Stabilization}. An extension from $n+1$ to $n+m$ coupled PDEs, with actuation on $m$ ``homodirectional'' PDEs, was provided in \cite{Hu2016Control,hu2019boundary}. Redesigns that are robust to delays were provided in \cite{Auriol2018Delay1}. 
An extension from coupled hyperbolic PDEs to cascades with ODEs was presented in \cite{DIMEGLIO2018281}. An extension from hyperbolic PDE-ODE cascades to ``sandwiched'' ODE-PDE-ODE systems was presented in \cite{WANG2020109131} and an event-triggered design for such systems was given in \cite{9319184}. The extension of PDE backstepping to output-feedback regulation with disturbances is proposed in \cite{DEUTSCHER201556,deutscher2018}. For coupled hyperbolic PDEs with unknown parameters, a comprehensive collection of adaptive control designs was provided in the book \cite{Anfinsen2019Adaptive}. Applications of backstepping to coupled hyperbolic PDE models of traffic are introduced in \cite{Yu2019Traffic,Yu2022}.

\paragraph{Paper outline and contributions}
After a brief introduction to the backstepping design in Section \ref{sec-bkst-intro}, for system \eqref{eq-PDE}, \eqref{eq-PDEBC}, in Section \ref{sec-kernelLip} we prove that the backstepping kernel operator is locally Lipschitz, between the spaces of continuous functions, with which we satisfy a sufficient condition for the existence of a neural operator approximation of a nonlinear operator to arbitrarily high accuracy---stated at the section's end in a formal result and illustrated with an example of approximating the operator $k={\cal K}(\beta)$. 
In Section \ref{sec-stabilityDeepONet} we present the first of our main results: the closed-loop stabilization (not merely practical but exponential) with a DeepONet-approximated backstepping gain kerne  l function. In Section \ref{sec-simulations} we present simulation results that illustrate stabilization under DeepONet-approximated gains. Then, in Section~\ref{sec-beta,u->u} we pose the question of whether we can not only approximate the gain kernel mapping $\beta(x)\mapsto k(x)$, as in Sections \ref{sec-kernelLip}
and \ref{sec-stabilityDeepONet}, but the entire feedback law mapping $(\beta(x),u(x,t)) \mapsto \int_0^1 k(1-y)u(y,t)dy$ at each time instant $t$; we provide an affirmative answer and a guarantee of semiglobal practical exponential stability under such a DeepONet approximation. In Section \ref{sims-fbklawapprox} we illustrate this feedback law approximation with a theory-confirming simulation. Then, in Section \ref{sec-PIDEextension}, we present the paper's most general result, which we leave for the end for pedagogical reasons, since it deals with Volterra operator kernel functions of two variables, $(x,y)$, on a triangular domain, and requires continuity of mappings between  spaces of functions that are not just continuous but continuously differentiable, so that not only the backstepping kernel is accurately approximable but also the kernel's spatial derivatives, as required for closed-loop stability. We close with a numerical illustration for this general case in Section \ref{sec-simulationsQ}.

In summary, the paper's contributions are the  PDE stabilization under DeepONet approximations of backstepping gain kernels (Theorems \ref{thm-stabDeepONet} and \ref{thm-stabDeepONet-gf}) and under the approximation of backstepping feedback laws (Theorem \ref{thm-semiglobalpractical}). Our stabilization results also hold for any  other neural operators with a universal approximation property (shown for LOCA~\cite{kissas2022loca} and for FNO on the periodic domain~\cite{kovachki2021universal}).   

\paragraph{Notation} 
We denote 
convolution operations 
as
\begin{eqnarray}
   ( a * b)(x) = \int_0^x a(x-y) b(y) dy
\end{eqnarray}
In the sequel, we 
suppresses the arguments $x$ and $t$ wherever 
clear from the context. For instance, we write \eqref{eq-PDE}, \eqref{eq-PDEBC} compactly as $u_t=u_x +\beta u(0)$ and $u(1)=U$, where, from the context, the boundary values $u(0), u(1)$ depend on $t$ as well. 

\section{Backstepping Design for a Transport PDE with `Recirculation'}
\label{sec-bkst-intro}

Consider the PDE system \eqref{eq-PDE}, \eqref{eq-PDEBC}. 
We employ the following backstepping transformation:
\begin{eqnarray}\label{eq-bkrsttransconv}
	w = u - k\ast u,
\end{eqnarray}
i.e., $w(x, t) = u(x, t) - \int_0^x k(x- y)u(y, t) dy$, to convert the plant into the target system
\begin{eqnarray}\label{eq-targetPDE}
	 w_t &= &w_x \\ \label{eq-targetPDEBC}
	w(1) &=& 0
\end{eqnarray}
with the help of feedback
\begin{equation}\label{sec-perfectfbk}
    U = (k\ast u)(1),
\end{equation}
namely, $U(t) = \int_0^1 k(1- y)u(y, t) dy$. 
To yield the target system, 
$k$ must satisfy the integral/convolution equation 
\begin{eqnarray} \label{eq:1.11}
	k(x) = - \beta(x)+ \int_0^x 
 \beta(x-y) k(y) 
 dy 
\end{eqnarray}
for $x\in[0,1]$. Note that, while this integral equation is linear in $k$ for a given $\beta$, the mapping from $\beta$ to $k$ is actually nonlinear, due to the product in the convolution of $\beta$ with $k$.

\section{Accuracy of Approximation of Backstepping Kernel Operator with DeepONet}
\label{sec-kernelLip}


An $n$-layer NN $f^\mathcal{N}:\mathbb{R}^{d_1}\rightarrow \mathbb{R}^{d_n}$ is given by
\begin{eqnarray} \label{eq-nn}
    f^\mathcal{N}(x, \theta) := (l_n \circ l_{n-1} \circ ... \circ l_2 \circ l_1) (x,\theta)
\end{eqnarray}
where layers $l_i$ start with $l_0 = x\in\mathbb{R}^{d_1}$ and continue as 
\begin{equation}
    l_{i+1}
    (l_i,\theta_{i+1}):= \sigma(W_{i+1} l_i + b_{i+1}), \quad i=1,\ldots,n-1
\end{equation}
$\sigma$ is a nonlinear activation function, and weights $W_{i+1} \in \mathbb{R}^{d_{i+1} \times d_i}$ and biases $b_{i+1} \in \mathbb{R}^{d_{i+1}}$ are parameters to be learned, collected into $\theta_i\in \mathbb{R}^{d_{i+1}(d_i+1)}$, 
and then 
into $\theta = [\theta_1^{\rm T},\ldots, \theta_n^{\rm T}]^{\rm T}\in\mathbb{R}^{\sum_{i=1}^{n-1} d_{i+1}(d_i+1)}$.
Let $\vartheta^{(k)}, \theta^{(k)} \in\mathbb{R}^{\sum_{i=1}^{k-1} d_{k,(i+1)}(d_{k,i}+1)}$ denote a sequence of NN weights.

An neural operator (NO) 
for approximating a nonlinear operator $\mathcal{G}: {\cal U} \mapsto {\cal V}$ is defined as 
\begin{eqnarray}
    {\cal G}_\mathbb{N}(\mathbf{u}_m)(y) = \sum_{k=1}^p g^\mathcal{N} (\mathbf{u}_m; \vartheta^{(k)}) f^\mathcal{N} (y; \theta^{(k)})
\end{eqnarray}
where ${\cal U}, {\cal V}$ are function spaces of continuous functions $u \in {\cal U}, v \in {\cal V}$.
$\mathbf{u}_m$ is the evaluation of function $u$ at points $x_i=x_1, ..., x_m$, $p$ is the number of chosen basis components in the target space, $y \in Y$ is the location of the output function $v(y)$ evaluations, 
and $g^\mathcal{N}$, $f^\mathcal{N}$ are NNs termed  branch and trunk networks. Note, $g^\mathcal{N}$ and $f^\mathcal{N}$ are not limited to feedforward NNs \ref{eq-nn}, but can also be of 
convolutional or recurrent.


\begin{theorem}\label{thm-DeepONet}
{\em (DeepONet universal approximation theorem \cite[Theorem 2.1]{lu2021advectionDeepONet}).}
Let 
$X \subset \mathbb{R}^{d_x}$ and 
 $Y \subset \mathbb{R}^{d_y}$ be compact sets of vectors $x\in X$ and $y\in Y$, respectively. 
Let ${\cal U}:X\rightarrow U\subset \mathbb{R}^{d_u}$ and ${\cal V}:Y\rightarrow V\subset \mathbb{R}^{d_v}$ be sets of continuous functions $u(x)$ and $v(y)$, respectively. 
Let ${\cal U}$ be also compact.
Assume  
the operator $\mathcal{G}: {\cal U} \rightarrow {\cal V}$ 
is 
continuous.
Then, for all $\epsilon > 0$, there exist $m^*, p^* \in \mathbb{N}$ such that for each $m \geq m^*$, $p \geq p^*$, there exist $\theta^{(k)}, \vartheta^{(k)}$, neural networks $f^{\mathcal{N}}(\cdot ; \theta^{(k)}) , g^\mathcal{N}(\cdot ; \vartheta^{(k)}),  k=1,\ldots, p$, and $ x_j \in X,  j=1, \ldots, m$, with corresponding $\mathbf{u}_m = (u(x_1), u(x_2), \cdots, u(x_m))^{\rm T}$, such that
\begin{equation} \label{eq:2.14G}
|\mathcal{G}(u)(y) - \mathcal{G}_\mathbb{N}(\mathbf{u}_m)(y)| < \epsilon 
\end{equation}
for all functions $u \in {\cal U}$ and all values $y \in Y$ of ${\cal G}(u)\in{\cal V}$. 
\end{theorem}

\begin{definition}{\em (backstepping kernel operator).}
\label{def:kernel}
A mapping ${\cal K}:\beta \mapsto k$ of $C^0[0,1]$ into itself, where $k={\cal K} (\beta)$
satisfies 
\begin{eqnarray}\label{eq:2.12}
\mathcal{K}(\beta) = -\beta + \beta \ast\mathcal{K}(\beta),
\end{eqnarray}
namely, in the Laplace transform notation, 
\begin{equation}
     k=\mathcal{K}(\beta) := \mathscr{L}^{-1} \left\{\frac{\mathscr{L}\{ \beta\}}{\mathscr{L}\{\beta\} - 1}\right\}
\end{equation}
is referred to as the {\em backstepping kernel operator}.
\end{definition}

\begin{lemma}{\em (Lipschitzness of backstepping kernel operator ${\cal K}$).}\label{lemma:holder}
The kernel operator $\mathcal{K}:\beta \mapsto k$  in Definition \ref{def:kernel} is Lipschitz. Specifically, for any $B>0$ the operator $\mathcal{K}$ satisfies 
\begin{eqnarray}
        ||\mathcal{K}(\beta_1) - \mathcal{K}(\beta_2) ||_\infty \leq C ||\beta_1 - \beta_2||_\infty
\end{eqnarray}
with the Lipschitz constant
\begin{equation}\label{eq-LipC}
C=
{\rm e}^{3B}
\end{equation}
for any pair of functions $(\beta_1, \beta_2)$ such that $\|\beta_1\|_\infty, \|\beta_2\|_\infty \leq B$, where $\|\cdot\|_\infty$ is the supremum norm over the argument of $\beta$ and $k$. 
\end{lemma}

\begin{proof} 
Start with the iteration $k^0 = -\beta, k^{n+1} =k^0 + \beta\ast k^{n}, \ n\geq 0 $ and consider the iteration
\begin{equation}\label{eq-Deltakn}
\Delta k^{n+1} =\beta\ast \Delta k^n,
 \qquad \Delta k^0 = k^0 = - \beta
\end{equation}
for the difference $\Delta k^n = k^n-k^{n-1}$, which sums to 
\begin{eqnarray}\label{eq-Deltaknsum}
    k = \sum_{n=1}^\infty \Delta k^n.
  \end{eqnarray}
Next, for $\bar\beta = \|\beta\|_\infty$ and all $x\in[0,1]$, 
\begin{eqnarray}\label{eq-Deltaknbound}
    \left|\Delta k^{n}(x) \right| \leq {\bar\beta^{n+1}x^{n}\over n!},
\end{eqnarray}
which is established by induction by postulating $ \left|\Delta k^{n-1}(x) \right| \leq {\bar\beta^{n}x^{n-1}\over (n-1)!}$ and by computing, from \eqref{eq-Deltakn}, 
\begin{eqnarray}\label{eq-Deltakninduction}
\left|\Delta k^{n}(x) \right| 
&=& \left| \int_0^x\beta (x-y) \Delta k^{n-1}(y)dy\right|
\nonumber \\ &\leq&  \bar\beta \int_0^x {\bar\beta^{n}y^{n-1}\over (n-1)!}|dy
\leq {\bar\beta^{n+1}x^{n}\over n!}.
\end{eqnarray}
And then, \eqref{eq-Deltaknbound} and \eqref{eq-Deltaknsum} yield
\begin{equation}\label{eq-kbound}
    |k(x)|\leq \bar\beta {\rm e}^{\bar\beta x}
    . 
\end{equation}
Next, for $k_1 = {\cal K}(\beta_1)$ and $k_2 = {\cal K}(\beta_2)$ it is easily verified that 
\begin{equation}
    k_1 - k_2 = \beta_1 \ast(k_1-k_2) -\delta\beta +\delta\beta\ast k_2
\end{equation}
where $\delta\beta = \beta_1 - \beta_2$. Define the iteration
\begin{eqnarray}\label{eq-deltakndelta}
    \delta k^{n+1} &=& \beta_1\ast \delta k^n
    \\ \label{eq-deltak0}
    \delta k^0 &=& - \delta\beta + \delta\beta \ast k_2
\end{eqnarray}
which verifies
$    k_1 -
    k_2 = \sum_{n=1}^\infty \delta k^n$.
Noting that \eqref{eq-kbound} ensures that $k_2={\cal K}(\beta_2)$ verifies $|k_2(x)|\leq \bar\beta_2 {\rm e}^{\bar\beta_2 x}$, from \eqref{eq-deltak0}, 
\begin{equation}
    |\delta k^0 (x)| \leq \left(1+\bar\beta_2{\rm e}^{\bar\beta_2 x}\right)\overline{\delta\beta} \leq \mu_2 \overline{\delta\beta}
\end{equation}
where 
$   \mu_2 := 1+\bar\beta_2{\rm e}^{\bar\beta_2 }$ and $ \overline{\delta\beta} = \|\beta_1-\beta_2\|_\infty$, 
it can be shown by induction, by mimicking the chain of inequalities \eqref{eq-Deltakninduction}, that, for all $x\in[0,1]$, 
\begin{equation}
    \left|\delta k^n(x)\right| \leq \mu_2\overline{\delta\beta} {\bar\beta_1^{n}x^{n}\over n!} 
\end{equation}
and therefore it follows that, for all $x\in[0,1]$, 
\begin{eqnarray}
    |k_1(x) - k_2(x)|& \leq &\left(1+\bar\beta_2{\rm e}^{\bar\beta_2}\right) {\rm e}^{\bar\beta_1 x} \|\beta_1-\beta_2\|_\infty 
    \nonumber \\
    & \leq & {\rm e}^{3B } \|\beta_1-\beta_2\|_\infty  . 
\end{eqnarray}
Hence, local 
Lipschitzness is proven with 
\eqref{eq-LipC}.
\end{proof}

\begin{corollary}{\em (to Theorem~\ref{thm-DeepONet}).}\label{cor-DeepONet}
Consider the backstepping kernel operator ${\cal K}$ in Definition \ref{def:kernel}. 
For all $B>0$ and $\epsilon > 0$, there exist $p^*(B,\epsilon), m^*(B,\epsilon) \in \mathbb{N}$, with an increasing dependence on $B$ and $1/\epsilon$, such that for each $p \geq p^*$ and $m \geq m^*$ there exist $\theta^{(k)}, \vartheta^{(k)}$, neural networks $f^{\mathcal{N}}(\cdot ; \theta^{(k)}) , g^\mathcal{N}(\cdot ; \vartheta^{(k)}),  k=1,\ldots, p$, and $ x_j \in K_1, j=1, \ldots, m$, with corresponding $\mathbf{\beta}_m = (\beta(x_1), \beta(x_2), \cdots, \beta(x_m))^{\rm T}$,
such that
\begin{equation} \label{eq:2.14}
|\mathcal{K}(\beta )(x) - \mathcal{K}_\mathbb{N}(\mathbf{\beta}_m)(x)| < \epsilon 
\end{equation}
holds for all Lipschitz $\beta$ with the property that
$\|\beta\|_\infty\leq B$. 
\end{corollary}

\begin{figure*}[t]
    \includegraphics{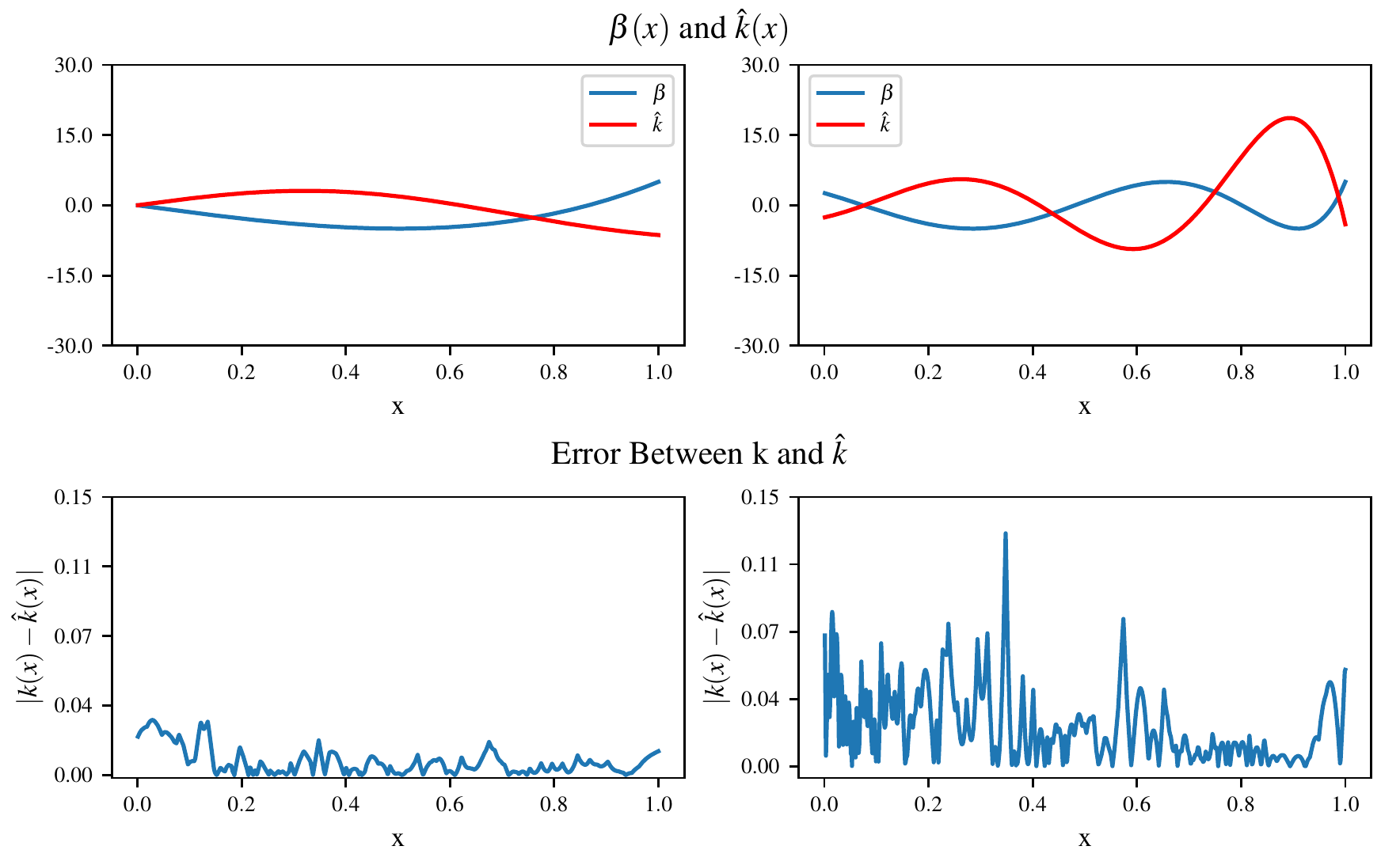}
	\caption{Examples of $\beta$, $\hat{k}$ for Chebyshev polynomials defined as $\beta = 6 \cos (\gamma \cos^{-1}(x))$ with $\gamma = 3, 7.35$ on the left and right respectively. The $\gamma$ parameter controls the wave frequency of $\beta$ and therefore affects the resulting kernel. Additionally, the DeepONet absolute approximation error of $\hat{k}$ and $k$ is shown. The DeepONet approximates the "smoother" function on the left with better precision than the large, oscillating function on the right. }
    \label{fig:1}
\end{figure*}

So the backstepping kernel is approximable, qualitatively, but how many neurons and how much data are needed for a given $\epsilon$? We recall a result on the minimum-sized DeepONet. 

\begin{proposition}{\em (DeepONet size for kernel operator approximation {\cite[Theorem 3.3 and Remark 3.4]{lu2021advectionDeepONet}}).}
    If the  kernel operator defined in \eqref{eq:2.12} is Lipschitz (or at least H\"{o}lder) continuous, a  DeepONet that approximates it to a {\em required} error tolerance $\epsilon >0$ indicated by \eqref{eq:2.14}
    employs the number of data point evaluations for $\beta$ on the order of
    \begin{equation}
        m \sim \epsilon^{-1},
    \end{equation} 
    the number of basis components in the interpolation when reconstructing into $C^0[0,1]$ on the order of 
    \begin{equation}
        p \sim \epsilon^{-\frac{1}{2}},
    \end{equation}
    the numbers of layers $L_{g^\mathcal{N}}$ in the branch network and of neurons $N_{g^\mathcal{N}}$ in each layer of the branch network on the order given, respectively, by 
    \begin{eqnarray}
        N_{g^\mathcal{N}} \cdot L_{g^\mathcal{N}} \sim \left({1\over\epsilon}\right)^{\frac{1}{\epsilon}},
    \end{eqnarray}  
    and the total size of the trunk network on the order of
    \begin{eqnarray}
        |\theta^{(k)}| \sim \left(\frac{3}{2} \log \frac{1}{\epsilon}\right)^2. 
    \end{eqnarray}

\end{proposition}  

\begin{example} \label{example1}
In Figure \ref{fig:1} we present two examples of approximation of 
$k$ using a DeepONet approximation of ${\cal K}(\beta)$ for given $\beta_1$ and $\beta_2$, 
which are taken as  Chebyshev polynomials  $\beta(x) = 6 \cos (\gamma \cos^{-1}(x))$.
They are trained on approximating kernels from 900 samples with $\gamma \in \text{uniform}[2, 8]$.
\end{example}

\section{Stability under Kernel Approximation with DeepONet}
\label{sec-stabilityDeepONet}

For our stability study under an approximate (imperfect) kernel, we begin with a derivation of the target PDE system under a backstepping transformation employing a DeepONet approximation of the backstepping kernel. 
 
For a given $\beta$, let $\hat k = \hat{\mathcal{K}} (\beta)$, where $\hat{\mathcal{K}} = \mathcal{K}_\mathbb{N}$, denote an NO approximation of the exact backstepping kernel $k$ whose existence is established in Corollary \ref{cor-DeepONet} for DeepONet. Let 
\begin{equation}
    \tilde k = k-\hat k
\end{equation}
denote the approximation error. Finally, let the backstepping transformation with the approximate kernel $\hat k$ be 
\begin{eqnarray}\label{eq-bksttrans}
    \hat{w} = u - \hat{k} * u.
\end{eqnarray}

With routine calculations, employing the approximate backstepping transformation and the feedback 
\begin{equation}\label{eq-khatfbk}
    U = (\hat k \ast u) (1)
\end{equation}
we arrive at the target system 
\begin{eqnarray}\label{eq-what-target1}
    \hat{w}_t &=& \hat{w}_x + \delta\hat{w}(0)
    \\  \label{eq-what-target2}
\hat w(1) &=&0, 
\end{eqnarray}
where the function $\delta(x)$ is defined as
\begin{equation}\label{eq:2.30}
    \delta = -\tilde k + \beta\ast\tilde k. 
\end{equation}


Next, we proceed with a Lyapunov analysis.

\begin{lemma}{\em (a Lyapunov estimate).}
\label{lem-Lyap}
Given arbitrarily large $B>0$, for all Lipschitz $\beta$ with $\|\beta\|_\infty \leq B$, and for all neural operators $\hat{\cal K}$ with $\epsilon \in (0, \epsilon^*)$, where
\begin{equation} \label{eq-eps*}
    \epsilon^*(B) = {c {\rm e}^{-c/2}\over 1+B}
\end{equation}
the Lyapunov functional
\begin{equation}
    V(t) = \int_0^1 {\rm e}^{cx} \hat w^2(x,t) dx, \qquad c>0.
\end{equation}
satisfies the following estimate along the solutions of the target system \eqref{eq-what-target1}, \eqref{eq-what-target2},
\begin{equation}
    V(t) \leq V(0) {\rm e}^{-c^*t}, 
\end{equation}
for  
\begin{equation}\label{eq-c*}
    c^* =c- \frac{e^c}{c}\epsilon^2 \left(1+ B \right)^2 >0. 
\end{equation} 
The  accuracy required of the NO $\hat{\cal K}$, and given by \eqref{eq-eps*}, is maximized with $c=2$ and has the value $\epsilon^*(B) = \frac{2}{{\rm e}\left(1+B\right)}$.
\end{lemma}

\begin{proof}
Several steps of calculation (chain rule, substitution, integration by parts) result in 
\begin{eqnarray}
    \dot V &=& - \hat w^2(0) - c \int_0^1 {\rm e}^{cx} \hat w^2(x,t) dx 
    \nonumber \\
    &&
    + \hat w(0) \int_0^1 \delta(x) {\rm e}^{cx} \hat w(x) dx
    \nonumber \\
    &\leq &
    -{1\over 2 }w^2(0) - c \int_0^1 {\rm e}^{cx} \hat w^2(x,t) dx 
    \nonumber\\
    &&
    + \left(\int_0^1 \delta(x) {\rm e}^{cx} \hat w(x) dx\right)^2
\end{eqnarray}
With the Cauchy-Schwartz inequality
\begin{eqnarray}
    &&
    \left(\int_0^1 \delta(x) {\rm e}^{cx} \hat w(x) dx\right)^2 \nonumber \\ 
    &&
    \leq \int_0^1 \delta^2(x) {\rm e}^{cx} dx\int_0^1 {\rm e}^{cx} \hat w(x)^2 dx
\end{eqnarray}
we get
\begin{eqnarray}
    \dot V 
    & \leq &
    -{1\over 2 }w^2(0) - \left(c- \int_0^1 \delta^2(x) {\rm e}^{cx} dx\right)V
\end{eqnarray}
The function $\delta$ in \eqref{eq:2.30} is bounded by
$|\delta(x)| \leq \left(1+||\beta||_\infty\right)||\tilde{k}||_\infty  $ 
which, in turn, using \eqref{eq:2.14}, yields 
\begin{equation}\label{eq-deltabar}
    |\delta(x)| \leq (1+  \bar{\beta})\epsilon =: \bar\delta.
\end{equation}
Then, substituting this into (37), we obtain:
\begin{eqnarray}
    \dot V 
    &\leq &
    -{1\over 2 }w^2(0) - \left(c- \epsilon^2\left(1 +  \bar{\beta}\right)^2\int_0^1  {\rm e}^{cx} dx\right)V 
    \nonumber\\ 
     &\leq &
    -{1\over 2 }w^2(0) - \left(c- \frac{e^c}{c}\epsilon^2 \left(1+ \bar{\beta} \right)^2 \right)V
     \nonumber\\ 
     &\leq &
    -{1\over 2 }w^2(0) - \left(c- \frac{e^c}{c}\epsilon^2 \left(1+ B \right)^2 \right)V
\end{eqnarray}
For $0\leq \epsilon \leq \epsilon^*$, where $\epsilon^*$ is defined in \eqref{eq-eps*}, 
we have
\begin{eqnarray}
    \dot V 
    \leq 
    -{1\over 2 }w^2(0) - c^* V 
\end{eqnarray}
for some $c^*>0$ in \eqref{eq-c*}. 
\end{proof}

The size of the NO and of the dataset needs to increase with $\bar\beta$, i.e., with the potential instability in the open-loop system. 

\begin{lemma}{\em (bound on inverse approximate kernel).}
\label{eq-Lyapsandwich}
The kernel $\hat l$ of the inverse to the  backstepping transformation \eqref{eq-bksttrans},
\begin{eqnarray}
    u = \hat{w} + \hat{l} \ast \hat w,
\end{eqnarray}
satisfies, for all $x\in[0,1]$, the estimate
\begin{equation}\label{eq-lhatestimate}
    |\hat l(x)| \leq \left(\bar\beta + (1+  \bar{\beta})\epsilon \right) {\rm e}^{(1+  \bar{\beta})\epsilon  x}.
\end{equation}
\end{lemma}

\begin{proof}
It is easily shown that $\hat l$ obeys the integral equation
\begin{equation}
    \hat l = -\beta +\delta + \delta\ast\hat l.
\end{equation}
Using the successive approximation approach, we get that the following bound
holds for all $x\in[0,1]$:
\begin{equation}
    |\hat l(x)| \leq \left(\bar\beta + \bar\delta\right) {\rm e}^{\bar\delta x}.
\end{equation}
With \eqref{eq-deltabar}, we  get \eqref{eq-lhatestimate}.
\end{proof}

\begin{theorem}{\em (Closed-loop stability robust to DeepONet approximation of backstepping kernel).}
\label{thm-stabDeepONet}
Let $B>0$ be arbitrarily large and consider the closed-loop system consisting of \eqref{eq-PDE}, \eqref{eq-PDEBC} with any Lipschitz $\beta$ such that $\|\beta\|_\infty \leq B$, and the feedback \eqref{eq-khatfbk} with the NO gain kernel $\hat k = \hat{\cal K}(\beta)$ of arbitrary desired accuracy of approximation $\epsilon\in (0,\epsilon^*)$ in relation to the exact backstepping kernel $k$, where $\epsilon^*(B)$ is defined in \eqref{eq-eps*}. This closed-loop system obeys the exponential stability estimate
   \begin{equation}\label{eq-esestimate}
    \|u(t)\| \leq M {\rm e}^{-c^*t/2} \|u(0)\|, \qquad \forall t\geq 0
\end{equation}
with the overshoot coefficient
\begin{equation}
 M=\left(1+\left(\bar\beta + (1+  \bar{\beta})\epsilon \right) {\rm e}^{(1+  \bar{\beta})\epsilon } \right) \left(1+\bar\beta {\rm e}^{\bar\beta} \right) {\rm e}^{c/2}.
\end{equation}
\end{theorem}

\begin{proof}
First, we note that $V$ 
Lemma \ref{lem-Lyap} satisfies
\begin{equation}\label{eq-Vsandwich}
    {1\over \left(1+\|\hat l\|_\infty \right)^2}\|u\|^2 \leq V \leq {\rm e}^c \left(1+\|\hat k\|_\infty \right)^2 \|u\|^2. 
\end{equation}
Since, by Lemma~\ref{lem-Lyap}, $V(t) \leq V(0) {\rm e}^{-c^*t}$, we get, for all $t\geq 0$, 
\begin{eqnarray}
    \|u(t)\| &\leq& \left(1+\|\hat l\|_\infty \right) \left(1+\|\hat k\|_\infty \right){\rm e}^{c/2}
    \nonumber \\
    && 
    \times  {\rm e}^{-c^*t/2} \|u(0)\|.
\end{eqnarray}
Then, noting, with Theorem \ref{thm-DeepONet}, \eqref{eq-kbound}, and Lemma \ref{eq-Lyapsandwich} that 
\begin{eqnarray}\label{eq-khatinf}
    \|\hat k\|_\infty &\leq & \|k\|_\infty +\epsilon \leq \bar\beta {\rm e}^{\bar\beta} +\epsilon
    \\  \label{eq-lhatinf}
    \|\hat l\|_\infty &\leq & \left(\bar\beta + (1+  \bar{\beta})\epsilon \right) {\rm e}^{(1+  \bar{\beta})\epsilon }
\end{eqnarray}
we finally arrive at the exponential stability estimate \eqref{eq-esestimate}. 
\end{proof}

\begin{remark}
\label{rem-observer}
Full-state measurement $u(x,t)$ is employed in the feedback law \eqref{eq-khatfbk} but can be avoided by employing only the measurement of the outlet signal, $u(0,t)$, from which the full state $u(x,t)$ is observable, the observer
\begin{eqnarray}\label{exp-obsPDE}
    \breve u_t &=& \breve u_x +\beta u(0)\\
    \label{exp-obsPDEBC}
    \hat u(1) &=& U
\end{eqnarray}
and the observer-based controller
\begin{equation}
    U=(\hat k \ast \breve u)(1),
\end{equation}
which can avoid solving the PDE \eqref{exp-obsPDE}, \eqref{exp-obsPDEBC} online by employing its explicit solution as an arbitrary function $\breve u(x,t) = \breve u_0(x)$ for $t+x \in [0,1)$ and
\begin{equation}
\breve u(x,t) = U(t+x-1) + \int_{t+x-1}^t \beta(t+x-\tau) u(0,\tau) d\tau
\end{equation}
for $t+x\geq 1$.
A closed-loop stability result as in Theorem \ref{thm-stabDeepONet} can be established for this observer-based controller. 
\end{remark}

\begin{figure*}[t]
    \includegraphics[width=\textwidth]{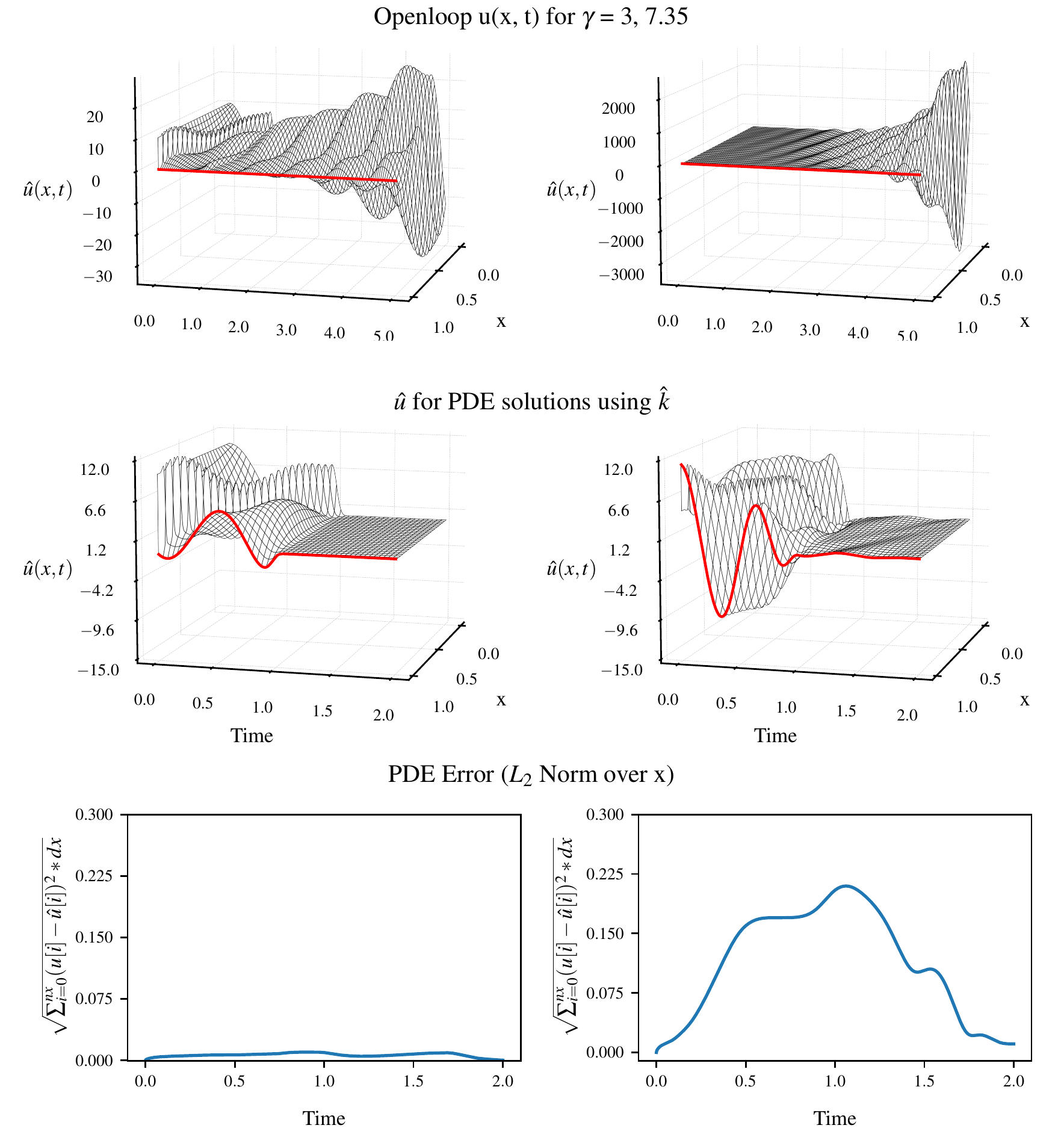} 
	\caption{Top row showcases open-loop instability for the recirculation functions  $\beta$ that are the same as in Fig. \ref{fig:1}, with $\gamma=3, 7.35$ on the left and right respectively. Additionally, the bottom two rows highlight examples of PDE closed-loop state response and errors between the response with ``perfect gain'' $k$ and ``approximate gain'' $\hat{k}$. $\beta$ corresponds to the same values in Figure \ref{fig:1}. For the more ``fluctuating'' plant parameter $\beta$, on the right of Figure \ref{fig:1}, the control task is more challenging and, consequently, the state approximation error is also higher (bottom right). }
    \label{fig:3}
\end{figure*}

\section{Simulations: Stabilization with NO-Approximated Gain Kernel $\beta\mapsto {\cal K}(\beta)$}
\label{sec-simulations}

Continuing with Example \ref{example1}, in Figure \ref{fig:3} we show that the system is open-loop unstable for both $\beta$s and we present tests with the learned kernels in closed-loop simulations up to $t=2$. In both cases, the PDE settles (nearly perfectly) by  $t=1,$ as expected from the target system with the perfect kernel $k$. The small ripple in the right simulation is due to the use of the approximated kernel $\hat k$. The simulations confirm the theoretical  guarantee that an NO-approximated kernel can successfully emulate a backstepping kernel while maintaining stability. 

 The NO architecture in $\hat{\cal K}$ consists of about 680 thousand parameters with a training time of $1$ minute (using an Nvidia RTX 3090Ti GPU) on a dataset of $900$ different $\beta$ defined as the Chebyshev polynomials $\beta = 6 \cos (\gamma \cos^{-1}(x))$ where $\gamma \sim \text{uniform(2, 10)}$. We choose $\beta$ of this form due to the rich set of PDEs and kernel functions constructed by varying only a single parameter. The resulting training relative $L_2$ error $4e-3$ and the testing relative $L_2$ loss on $100$ instances sampled from the same distribution was $5e-3$. If a wider distribution of $\gamma$ is chosen, the mapping can be learned but requires both a larger network and more data for the same accuracy. 
 
\section{Approximating the Full Feedback Law Map $(\beta,u) \mapsto U$}
\label{sec-beta,u->u}

We have so far pursued only the  approximation of  operator ${\cal K}(\beta)$, while treating the feedback operator \eqref{sec-perfectfbk}, given by $U = (k\ast u)(1)=({\cal K}(\beta)\ast u)(1)$, as straightforward to compute---merely an integral in $x$, i.e., a simple inner product between the functions ${\cal K}(\beta)(1-x)$ and the state measurement $u(x,t)$. 

It is of theoretical (if not practical) interest
to explore the neural approximation of the mapping from $(\beta,u)$ into the scalar control input $U$. Such a mapping is clearly from a much larger space of functions $(\beta,u)$ into scalars (i.e., the mapping is {\em functional}) and is, therefore, considerably more training-intensive and learning-intensive. Nevertheless, since it is legitimate to ask how one would approximate not just the feedback gain kernel but the entire feedback law map, we examine this option in this section. 

We emphasize that we are approximating just the feedback operator $({\cal K}(\beta)\ast u)(1)$, whose second argument is the current state $u$ as a function of $x$, not the entire trajectory $u(x,t)$. 
We do not train the NO using a trajectory-dependent cost 
$\int_0^{t_{\rm f}} \left(\int_0^1 u^2(x,t) dx + U^2(t)\right) dt $ for different initial conditions $u_0$, as, e.g., in the application of RL to the hyperbolic PDEs of traffic flow in \cite{9568241}. Instead, we perform the training simply on the kernel integral equation \eqref{eq:2.12} and the convolution operation \eqref{sec-perfectfbk} for sample functions $\beta$ and $u$ of $x$. 

The form of stability we achieve in this section is less strong than in Theorem \ref{thm-stabDeepONet}. While Theorem \ref{thm-stabDeepONet} guarantees global exponential stability, here we achieve only {\em semiglobal practical} exponential stability. Because in this section we do not just train a multiplicative gain ${\cal K}(\beta)$ but a feedback of $u$ as well, the approximation error is not just multiplicative but additive, which is the cause of the exponential stability being {\em practical}. Because the data set involves samples $u$ of bounded magnitude, stability is {\em semiglobal} only. 

Nevertheless, in comparison to the training on closed-loop solutions over a finite time horizon for the traffic flow in \cite{9568241}, where the finite horizon precludes the possibility of stability guarantees, the semiglobal practical exponential stability achieved here is a rather strong result.

We start by establishing the Lipschitzness of the backstepping feedback map.

\begin{lemma}Consider the feedback \eqref{sec-perfectfbk}, namely, 
\begin{equation}\label{sec-perfectfbk+}
    U = ({\cal K}(\beta) \ast u)(1),
\end{equation}
and the associated map ${\cal U}: (\beta,u) \mapsto U$ from $C^0([0,1]^2)$ into $\mathbb{R}$. 
For arbitrary  $B_\beta,B_u>0$, the mapping ${\cal U}$ is Lipschitz on any set of $x$-dependent Lipschitz functions $(\beta,u)$ such that $\|\beta\|_\infty\leq B_\beta, \|u\|_\infty \leq B_u$, with a Lipschitz constant 
\begin{equation}
    C_{\cal U} = B_\beta {\rm e}^{B_\beta} + B_u {\rm e}^{3B_\beta}. 
\end{equation}
\end{lemma}

\begin{proof}
    Let $U_1 = {\cal U}(\beta_1,u_1) = ({\cal K}(\beta_1) \ast u_1)(1)$ and $U_2 = {\cal U}(\beta_2,u_2) = ({\cal K}(\beta_2) \ast u_2)(1)$. A  calculation gives
\begin{eqnarray}
   && |U_1-U_2| = |({\cal K}(\beta_1) \ast u_1)(1) - ({\cal K}(\beta_2) \ast u_2)(1)|
    \nonumber\\
    &&\leq \|{\cal K}(\beta_1)\|_\infty \|u_1-u_2\|_\infty
    + \|u_2\|_\infty \|{\cal K}(\beta_1)- {\cal K}(\beta_2) \|_\infty. 
    \nonumber \\
\end{eqnarray}
Let $\|\beta_1\|_\infty, \|\beta_2\|_\infty \leq B_\beta$ and $\|u_1\|_\infty, \|u_2\|_\infty \leq B_u$. Recall that $\|{\cal K}(\beta)\|_\infty \leq B_\beta {\rm e}^{B_\beta}$ and $\|{\cal K}(\beta_1)- {\cal K}(\beta_2) \|_\infty\leq {\rm e}^{3B_\beta}\|\beta_1-\beta_2\|_\infty$. Then we get 
\begin{eqnarray}
   && |{\cal U}(\beta_1,u_1) -{\cal U}(\beta_2,u_2)|
    \nonumber\\
    &&\leq 
    \left(B_\beta {\rm e}^{B_\beta} + B_u {\rm e}^{3B_\beta} \right)
    \|(\beta_1-\beta_2, u_1-u_2)\|_\infty. 
\end{eqnarray}
\end{proof}

Taking the backstepping transformation
$w = u - k\ast u$, where $k={\cal K}(\beta)$ is the exact backstepping kernel for $\beta$,
we get 
\begin{eqnarray}
	 w_t &= &w_x \\ 
	w(1) &=& U - ({\cal K}(\beta)\ast u)(1)
\end{eqnarray}
Let now $\hat{\cal U}$ be the NO version of the mapping ${\cal U}(\beta,u) = ({\cal K}(\beta)\ast u)(1)$. Taking the NO control $U=\hat{\cal U}(\beta,u)$, we obtain the boundary condition $w(1) = \hat{\cal U}(\beta,u) - ({\cal K}(\beta)\ast u)(1)$,  namely, the target system
\begin{eqnarray}
	 w_t &= &w_x \\ 
	w(1) &=& \hat{\cal U}(\beta,u) - {\cal U}(\beta,u)
\end{eqnarray}

Due to the Lipschitzness of ${\cal U}$, based on the DeepONet approximation accuracy theorem, we get the following. 

\begin{lemma}\label{lemma-Utilde}
    For all $B_\beta,B_u>0$ and  $\epsilon$, there exists an NO $\hat{\cal U}$ such that 
\begin{eqnarray}
|{\cal U}(\beta,u) - \hat{\cal U}(\beta,u)|<\epsilon
\end{eqnarray}
for all $\beta, u \in C^0[0,1]$ that are Lipschitz in $x$ and such that $\|\beta\|_\infty\leq B_\beta, \|u\|_\infty \leq B_u$. 
\end{lemma}

Next, we state and then prove the main result.
\begin{theorem}{\em (Semiglobal practical stability under DeepONet approximation of backstepping feedback law).}\label{thm-semiglobalpractical}
    If $\epsilon <\epsilon^*$, where
\begin{equation}
    \epsilon^*(B_\beta,B_u,c) := {\sqrt{c} B_u\over {\rm e}^{c/2}  \left(1+ B_{\beta}\right)}\, >0,
\end{equation}
and $\|u(0)\| \leq B_u^0$, where
\begin{equation}\label{em-ROA-Ucal}
    B_u^0(\epsilon,B_\beta,B_u,c):= {1\over 1+B_\beta {\rm e}^{B_\beta}}
    \left({B_u\over {\rm e}^{c/2}  \left(1+ B_{\beta}\right)}-{\epsilon\over\sqrt{c}}
    \right)\, >0,
\end{equation}
the closed-loop solutions under the NO approximation of the PDE backstepping feedback law, i.e., 
\begin{eqnarray}
    \label{eq-PDEUcal}
	u_t(x,t) &=& u_x(x,t) + \beta(x) u(0, t) 
 \\  \label{eq-PDEBCUcal}
	u(1, t) &=& \hat{\cal U}(\beta,u)(t)
\end{eqnarray}
satisfy the  {\em semiglobal practical exponential stability} estimate
   \begin{eqnarray}\label{eq-esestimateUcal-final}
    \|u(t)\| &\leq& 
    \left(1+B_\beta \right) \left(1+B_\beta {\rm e}^{B_\beta} \right) {\rm e}^{c/2}{\rm e}^{-c t/2} \|u(0)\|
    \nonumber\\ 
    && +\left(1+ B_{\beta}\right)
    {{\rm e}^{c/2}\over \sqrt{c}} \epsilon
    , \qquad \forall t\geq 0. 
\end{eqnarray}
\end{theorem}

The estimate \eqref{eq-esestimateUcal-final} is semiglobal because the radius $B_u^0$ of the ball of initial conditions in $L^2[0,1]$ is made arbitrarily large by increasing $B_u$, and by increasing, in accordance with the increase of $B_u$, the training set size and the number of NN nodes. Nevertheless, though semiglobal, the attraction radius $B_u^0$ in \eqref{em-ROA-Ucal} is much smaller than the magnitude $B_u$ of the samples of $u$ in the training set. 

The residual value, 
\begin{equation}
    \lim\sup_{t\rightarrow} \|u(t)\| \leq \left(1+ B_{\beta}\right)
    {{\rm e}^{c/2}\over \sqrt{c}} \epsilon
\end{equation}
    is made arbitrarily small by decreasing $\epsilon$, and by increasing, in accordance with the decrease of $\epsilon$, the training set size and the number of NN nodes. As the magnitude $B_\beta$ of the (potentially destabilizing) gain samples $\beta$ used for training grows, the residual error grows. 

\begin{proof}{\em (of Theorem \ref{thm-semiglobalpractical})}
To make the notation concise, denote $\tilde{\cal U} = {\cal U} - \hat{\cal U}$ and note that this mapping satisfies $|\tilde{\cal U}(\beta,u)|= |w(1)|\leq \epsilon$ for all $\|\beta\|_\infty\leq B_\beta, \|u\|_\infty \leq B_u$. Note also that $\tilde{\cal U}$ depends on $\epsilon, B_\beta, B_u$ through the number of training data and NO size. 
Consider now the Lyapunov functional $
    V(t) = \int_0^1 {\rm e}^{cx}  w^2(x,t) dx$.
Its derivative is
\begin{eqnarray}
    \dot V &=& {\rm e}^{c} w^2(1) -  w^2(0) - c \int_0^1 {\rm e}^{cx}  w^2(x,t) dx
    \nonumber \\
    &\leq &  - c V + {\rm e}^{c} w^2(1)
\end{eqnarray}
which yields
\begin{eqnarray}
        V(t) &\leq& V(0) {\rm e}^{-ct}
        + {{\rm e}^{c}\over c} \sup_{0\leq\tau\leq t} w^2(1,\tau)
        \nonumber\\
        &\leq& V(0) {\rm e}^{-ct}
        + {{\rm e}^{c}\over c} \sup_{0\leq\tau\leq t} \left(\tilde{\cal U}(\beta,u)(\tau)\right)^2.
\end{eqnarray}
Using the facts that
\begin{equation}\label{eq-Vsandwich-k}
    {1\over \left(1+\| l\|_\infty \right)^2}\|u\|^2 \leq V \leq {\rm e}^c \left(1+\| k\|_\infty \right)^2 \|u\|^2. 
\end{equation}
and $\|k\|_\infty, \|l\|_\infty \leq B_{\beta} {\rm e}^{B_{\beta}}$, $\|l\|_\infty \leq B_{\beta} $
we get 
   \begin{eqnarray}\label{eq-esestimateUcal}
    \|u(t)\| &\leq& 
    \left(1+B_\beta \right) \left(1+B_\beta {\rm e}^{B_\beta} \right) {\rm e}^{c/2}{\rm e}^{-c t/2} \|u(0)\|
    \nonumber\\ 
    && +\left(1+ B_{\beta}\right)
    {{\rm e}^{c/2}\over \sqrt{c}} \sup_{0\leq\tau\leq t} \left|\tilde{\cal U}(\beta,u)(\tau)\right|.
\end{eqnarray}
The conclusions of the theorem are directly deduced from this estimate and the bound $|\tilde{\cal U}|<\epsilon$ in Lemma \ref{lemma-Utilde}.
\end{proof}

The NO $\hat{\cal U}: (\beta,u) \mapsto U$ is  complex, and therefore computationally burdensome in real time. Why not instead precompute the neural operator $\hat{\cal K}: \beta\mapsto \hat k$ and also find a DeepONet $\hat\Omega$ approximation of the {\em bilinear} map $\Omega : (k,u) \mapsto U$, which is simply the convolution $\Omega(k,u)(t) = \int_0^1 k(1-x) u(x,t) dy$, and then compute just $\hat\Omega(\hat k,u)(t)$ in real time, after computing $\hat k = \hat{\cal K}(\beta)$ offline? This is certainly possible. Why haven't we developed the theory for this approach? Simply because the theory for such a ``composition-of-operators'' approach, for $\hat\Omega(\hat{\cal K}(\beta),u)$, would be hardly any different, but just notationally more involved, than the theory that we provide here for the one-shot neural operator $\hat{\cal U}(\beta,u)$.

\section{Simulations: Practical Stabilization with NO-Approximated Feedback Law $(\beta,u)\rightarrow U$}
\label{sims-fbklawapprox}
Learning the map $(\beta, u) \mapsto U$ is harder than $\beta \mapsto k$ due to the combination of two functions, $\beta$, and $u$. We can learn the mapping using a training set defined by $\beta$ as in Figure \ref{fig:1} with $\gamma \in \text{uniform}(2, 6)$ and  random values of $u$. We present results with the learned mapping in Figure \ref{fig:4} where the learned control contains significant error. Due to this, we see that the PDE in the right of Figure \ref{fig:4} contains a significant ripple past the time $T=1$ whereas the analytically controlled PDE is stabilized, as stipulated by the target system, by $T=1$. When compared to the operator approximation for gain kernel in Figure \ref{fig:3} left, the PDE error is at least twice as large confirming the theoretical results in Theorem \ref{thm-semiglobalpractical} and Theorem \ref{thm-stabDeepONet}. 

Furthermore, the network architecture, as presented in Figure \ref{fig:5} requires significant enhancement over a traditional DeepONet. To learn this mapping, we emulate the operator structure where the map $(\beta, u)$ requires two DeepONet layers for the integral operators adjoined with linear layers for the multiplicative operation. Additionally, to make the network feasible, we use a smaller spatial resolution than in Section \ref{sec-simulations} and a larger dataset. The dataset requires a combination of both $\beta$ and $u$ and thus consists of 50000 instances. Therefore a network of approximately 415 thousand parameters takes approximately 20 minutes to train. We achieved a training relative $L_2$ error of $7.2e-3$ and a testing relative $L_2$ error of $3.3e-2$. This demonstrates, to the practical user, that the map $(\beta, u)$ requires more training data and significant architectural enhancements boosting training time, yet the error in Figure \ref{fig:4} is larger compared to employing the learned map $\beta \mapsto k$. 

\begin{figure*}[t]
    \includegraphics[width=\textwidth]{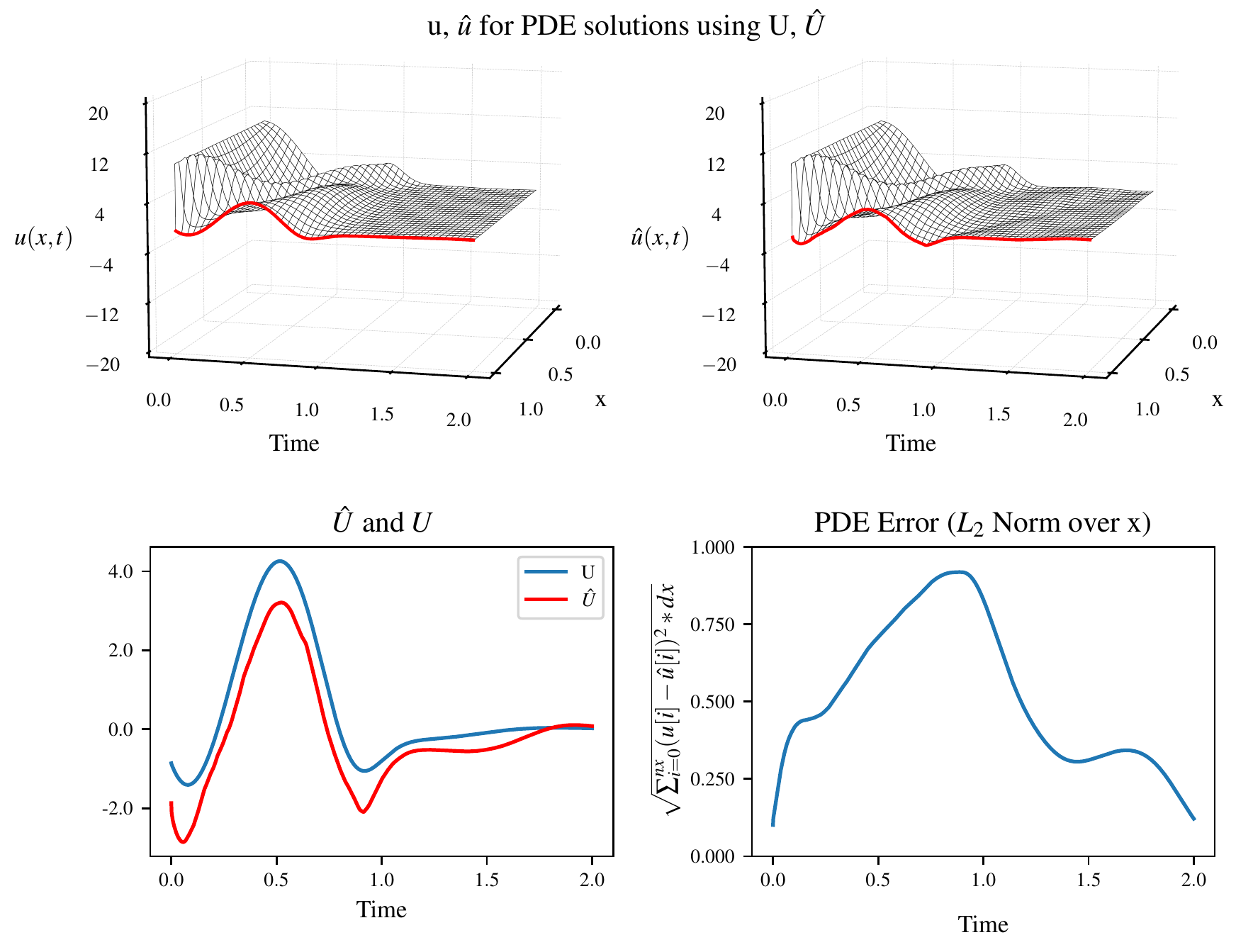} 
	\caption{Examples of PDE closed-loop state response and errors between the response with ``perfect control'' $U$ and ``approximate control'' $\hat{U}$. $\beta$ is same as in \ref{fig:1} with $\gamma=3$.}
    \label{fig:4}
\end{figure*}

\begin{figure*}[t]
    \includegraphics[width=\textwidth]{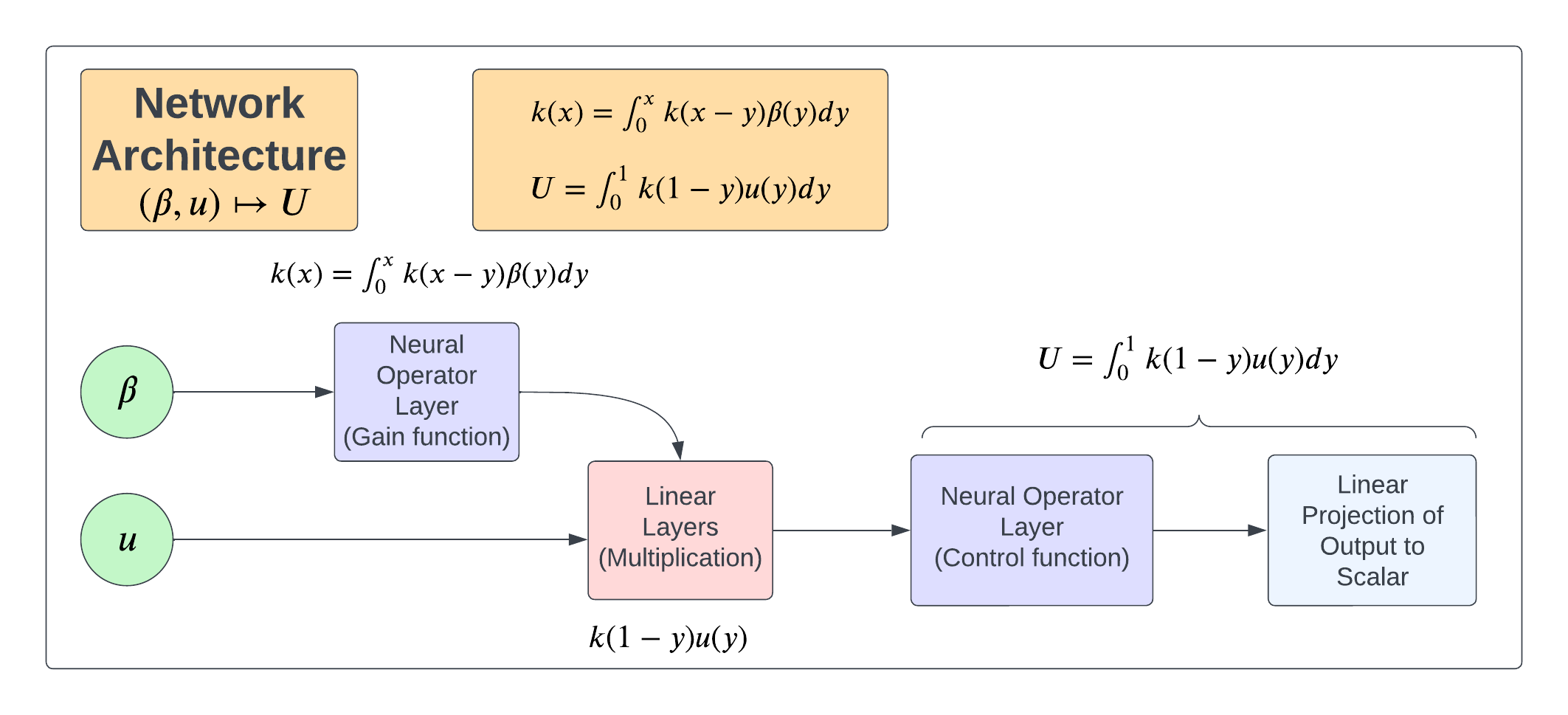} 
	\caption{Network architecture for the map $(\beta, u) \mapsto U$ presented in Section \ref{sec-beta,u->u}. The network first solves the kernel function using a DeepONet layer, then utilizes linear layers to multiply $k$ with the PDE state $u$, and concludes by learning a  second neural operator layer for the nonlinear integral operation yielding the final control output $U$. }
    \label{fig:5}
\end{figure*}

\section{Extension to Hyperbolic PIDEs}
\label{sec-PIDEextension}

We present the ``general case'' for a class of hyperbolic partial integro-differential equations (PIDE) of the form
\begin{eqnarray}
    \label{eq-PIDE}
	u_t(x,t) &=& u_x(x,t) + g(x) u(0,t) 
 \nonumber \\&&
 + \int_0^xf(x,y) u(y,t) dy , \qquad \mbox{$x\in[0,1)$}
 \\  \label{eq-PIDEBC}
	u(1, t) &=& U(t). 
\end{eqnarray}
We have left this generalization for the end of the paper for pedagogical reasons---in order not to overwhelm and daze the reader---since the case where the Volterra operator kernel $f(x,y)$ is a function of two variables complicates the treatment considerably. The backstepping transformation is no longer a convolution with a function of a single variable, the gain mapping is no longer of $C^0$ functions on $[0,1]$ but of $C^1$ functions on the triangle $\{0\leq y\leq x\leq 1\}$, and the DeepONet theorem requires estimates of the derivatives of the backstepping kernel. 

While \cite[(11)--(17)]{Bernard2014} shows that \eqref{eq-PDE}, \eqref{eq-PDEBC} can be transformed into \eqref{eq-PIDE}, \eqref{eq-PIDEBC}, this transformation involves a nonlinear mapping $f\mapsto \beta$, which itself would have to be learned to produce an approximation of the complete kernel mapping $(g,f)\mapsto k$ as a composition of two mappings. This is why the results of the previous sections do not provide a solution to the general case \eqref{eq-PIDE}, \eqref{eq-PIDEBC}, but only a pedagogical introduction, and this is why a generalization in this section is necessary.

To find the mapping from the PIDE coefficients $(g,f)$ to the kernel $k$ of the backstepping controller
\begin{equation}\label{eq-bkstfbkkPIDE}
    U(t) = \int_0^1 k(1,y) u(y,t) dy,
\end{equation}
we take the backstepping transform 
\begin{eqnarray}
 	w(x, t) = u(x, t) - \int_0^x k(x,y)u(y, t) dy,
\end{eqnarray}
which is not a simple convolution as in \eqref{eq-bkrsttransconv}, with a kernel depending on a single argument, and the same target system as in \eqref{eq-targetPDE}, \eqref{eq-targetPDEBC}, $
	 w_t = w_x, \ 
	w(1) = 0$, 
which gives the kernel integral equation derived in \cite{krstic2008Backstepping} as
\begin{eqnarray} \label{eq:kernelIntegral}
    k(x,y) &=& F_0(x,y)+F(g,f,k)(x,y),
\end{eqnarray}
where
\begin{align}\label{eq-F_0}
&F_0(x,y) := - g(x-y)- \int_0^y f(x-y+\xi,\xi) d\xi 
    \\  \label{eq-Fgf}
&F(g,f,\kappa)(x,y):= 
\int_0^{x-y} g(\xi) \kappa(x-y,\xi) d\xi
\nonumber\\
& +\int_0^y \int_0^{x-y} f(\xi+\eta,\eta)\kappa(x-y+\eta, \xi+\eta) d\xi d\eta.         
\end{align}

Denote ${\cal T} = \{0\leq y\leq x\leq 1\}$ as the domain of the functions $f$ and $k$. Further, denote 
\begin{align}
    \bar g = \sup_{[0,1]} |g|, \quad & \overline{g'} = \sup_{[0,1]} |g'|
    \\
    \bar f = \sup_{\cal T} |f|, \quad &\overline{f_x} = \sup_{\cal T} |f_x|. 
\end{align}
It was proven in \cite{krstic2008Backstepping}
that 
\begin{equation}\label{eq-kbarbound}
    |k(x,y)|\leq \left(\bar g + \bar f\right) {\rm e}^{\bar g + \bar f} =:\bar k\left(\bar g, \bar f\right).
\end{equation}
For the partial derivatives
\begin{eqnarray}
k_x &=&
F_0^x + F(g,f,k_x)\\
k_y &=&  F_0^y - F(g,f,k_x)
\end{eqnarray}
where
\begin{align}
&    F_0^x(x,y) =
-\int_0^y f_x(x-y+\xi, \xi) d\xi
+\phi_0(x,y)
\\
&F_0^y(x,y) = f_x(x,y) + \int_y^x f(\sigma, y) k(x, \sigma) d\sigma
-\phi_0(x,y)
\\
&    \phi_0(x,y) =
-g'(x-y)+g(x-y)k(x-y,x-y)
\nonumber\\
& + \int_0^y f(x-y+\eta, \eta) k(x-y+\eta, x-y+\eta)  d\eta
\end{align}
it is proven using the same approach (successive approximation, infinite series, induction) that, on the triangle ${\cal T}$, 
\begin{eqnarray}\label{eq-kxbarbound}
    |k_x(x,y)| &\leq& \left(\overline{f_x}+\overline{\phi_0}\right){\rm e}^{\bar g + \bar f} = :\overline{k_x}\left(\bar g, \overline{g'},\bar f, \overline{f_x}\right)
    \\  \label{eq-kybarbound}
    |k_y(x,y)| &\leq& \overline{f_x} + \bar f \bar k + \overline{\phi_0}+ \left(\bar g + \bar f\right) \overline{k_x}
\end{eqnarray}
where
\begin{equation}
    \overline{\phi_0}(\bar g, \overline{g'},\bar f) := \overline{g'}+
\left(\bar g + \bar f\right) \bar k.
\end{equation}
Hence, along with the existence, uniqueness, and continuous differentiability of $k$
\cite{krstic2008Backstepping}, we have proven the following.

\begin{lemma}
   The map ${\cal Q}:C^1([0,1]\times{\cal T})\rightarrow C^1([0,1]\times{\cal T})$ defined by $k = {\cal Q}(g,f)$, and representing the solution of \eqref{eq:kernelIntegral},   is continuous. In addition,  $|k|, |k_x|, |k_y|$ are bounded, respectively, as in \eqref{eq-kbarbound}, \eqref{eq-kxbarbound}, \eqref{eq-kybarbound}, in terms of the bounds on $|g|, |g'|, |f|, |f_x|$. 
\end{lemma}

From the continuity of the map ${\cal Q}$ on the Banach space $C^1([0,1]\times{\cal T})$, the following result is inferred from the DeepONet theorem. 

\begin{lemma}
For all $\epsilon>0$ and $B_g, B_{g'}, B_f, B_{f_x}>0$ there exists an NO $\hat{\cal Q}$ such that, for all $(x,y) \in {\cal T}$, 
 \begin{align}\label{eq-Qerrorbound}
 &\left|\hat{\cal Q}(g,f)(x,y) - {\cal Q}(g,f)(x,y)\right|
 \nonumber \\& 
 +\left|{\partial\over\partial y}\left(\hat{\cal Q}(g,f)(x,y) 
 - {\cal Q}(g,f)(x,y)\right)\right|
 \nonumber\\& 
 +\left|{\partial\over\partial y}\left(\hat{\cal Q}(g,f)(x,y) 
 - {\cal Q}(g,f)(x,y)\right)\right| < \epsilon 
 \end{align}
 for all functions $g\in C^1([0,1])$
 and $f\in C^1({\cal T})$ whose derivatives are Lipschitz and which satisfy $\|g\|_\infty \leq B_g$, $\|g'\|_\infty \leq B_{g'}$, $\|f\|_\infty \leq B_f$, $\|f_x\|_\infty \leq B_{f_x}$. 
\end{lemma}

Denoting $\tilde k = k-\hat k = {\cal K}(g,f) - \hat {\cal K}(g,f)$, \eqref{eq-Qerrorbound} can be written as $\left|\tilde k(x,y)\right|+\left|\tilde k_x(x,y)\right| + \left|\tilde k_y(x,y)\right|<\epsilon$. 

Now take the backstepping transformation 
\begin{equation}
    \hat w(x,t) = u(x,t) - \int_0^x \hat k(x,y) u(y,t) dy.
\end{equation}
With the control law
\begin{equation}\label{eq-fbkgf}
    U(t) = \int_0^1 \hat k(1,y) u(y,t) dy,
\end{equation}
the target system becomes
\begin{eqnarray}\label{eq-target-gf}
    \hat w_x(x,t) &=& \hat w_t(x,t) +\delta(x) \hat w(0,t) 
    \nonumber \\
    &&+ \int_0^x \delta_1(x,y) u(y,t) dy
    \\
    \hat w(1,t) &=&0,
\end{eqnarray}
where
\begin{eqnarray}
    \delta_0(x) &=& -\tilde k(x,0) + \int_0^x g(y) \tilde k(x,y) dy
    \\
    \delta_1(x,y) &=& -\tilde k_x(x,y) -\tilde k_y(x,y) 
    \nonumber\\
    && +\int_y^x f(\xi,y) \tilde k(x,\xi) d\xi
\end{eqnarray}
satisfy
\begin{eqnarray}
    \|\delta_0\|_\infty &\leq & (1+\bar g) \epsilon\\
    \|\delta_1\|_\infty &\leq & (2+\bar f) \epsilon.
\end{eqnarray}

Since the state $u$ appears under the integral in the $\hat w$-system \eqref{eq-target-gf}, in the Lyapunov analysis we need the inverse backstepping transformation 
\begin{equation}\label{eq-invbkstgf}
    u(x,t) = \hat w(x,t) + \int_0^x \hat l(x,y) \hat w(y,t) dy.
\end{equation}
It is shown in \cite{krstic2008boundary} that the direct and inverse backstepping kernels satisfy in general the relationship
\begin{equation}
\hat l(x,y) = \hat k(x,y) + \int_y^x \hat k(x,\xi) \hat l(\xi,y) dy.     
\end{equation}
The inverse kernel satisfies the following conservative bound
\begin{equation}
    \|\hat l\|_\infty \leq \|\hat k\|_\infty {\rm e}^{\|\hat k\|_\infty}.
\end{equation}
Since $\|k-\hat k\|_\infty <\epsilon$, we have that $\|\hat k\|_\infty \leq \|k\|_\infty+\epsilon$. With \eqref{eq-kbarbound} we get $\|\hat k\|_\infty \leq \bar k(\bar g, \bar f) +\epsilon$ and hence 
\begin{equation}
    \|\hat l\|_\infty \leq \left(\bar k +\epsilon\right) {\rm e}^{\bar k +\epsilon}.
\end{equation}
Going back to \eqref{eq-invbkstgf}, we get
\begin{equation}
    \|u\| \leq \left( 1+ \left(\bar k +\epsilon\right) {\rm e}^{\bar k +\epsilon}\right) \|\hat w\|.
\end{equation}
Mimicking and generalizing the steps of the proofs of Lemma \ref{lem-Lyap} and Theorem \ref{thm-stabDeepONet}, we get the following exponential stability result. (We omit the explicit but conservative and exceedingly complicated and uninformative estimates of the overshoot coefficient, the decay rate, and the upper bound $\epsilon^*$ on the approximation accuracy needed to guarantee stability under the gain approximation.)

\begin{theorem}
\label{thm-stabDeepONet-gf}
Let $B_g, B_{g'}, B_f, B_{f_x}>0$ be arbitrarily large and consider the system \eqref{eq-PIDE}, \eqref{eq-PIDEBC} with any $g\in C^1([0,1])$
 and $f\in C^1({\cal T})$
 whose derivatives are Lipschitz and which satisfy $\|g\|_\infty \leq B_g$, $\|g'\|_\infty \leq B_{g'}$, $\|f\|_\infty \leq B_f$, $\|f_x\|_\infty \leq B_{f_x}$. There exists a  sufficiently small $\epsilon^*(B_g, B_{g'}, B_f, B_{f_x})>0$ such that the feedback law \eqref{eq-fbkgf} with the NO gain kernel $\hat k = \hat{\cal Q}(g,f)$ of arbitrary desired accuracy of approximation $\epsilon\in (0,\epsilon^*)$ in relation to the exact backstepping kernel $k$ ensures that there exist $M, c^*>0$ such that the closed-loop system satisfies the exponential stability bound
   \begin{equation}\label{eq-esestimategf}
    \|u(t)\| \leq M {\rm e}^{-c^*t/2} \|u(0)\|, \qquad \forall t\geq 0. 
\end{equation}
\end{theorem}

\section{Simulations: Stabilization of PIDE with NO-Approximated Gain Kernel $f\mapsto {\cal Q}(f)$ Dependent on $(x,y)$}
\label{sec-simulationsQ}

For clarity, we consider the systems of the form \eqref{eq-PIDE} with $g=0$, so that the focus is solely on the mapping of two-dimensional plant kernels $f(x,y)$ into two-dimensional backstepping kernels $k(x,y)$, which are governed by the (double) integral equation
\begin{align} \label{eq:kernelIntegralf}
   & k(x,y) = - \int_0^y f(x-y+\xi,\xi) d\xi 
   \nonumber \\
&
     +  \int_0^y \int_0^{x-y} f(\xi+\eta,\eta)k(x-y+\eta, \xi+\eta) d\xi d\eta.
\end{align}
We illustrate in this section the NO approximation $\hat{\cal Q}$ of the nonlinear operator ${\cal Q}:f\mapsto k$ mapping $C^1({\cal T})$ into itself. First, in Figure \ref{fig:7} we present the construction of the two-dimensional function $f$ via a product of Chebyshev polynomials and highlight the PDE's open-loop instability. Then, we showcase the corresponding learned kernel and the error in Figure \ref{fig:8}. The pointwise error for the learned kernel peaks at around $10\%$ of $k$, as it ``ripples'' in the right of Figure \ref{fig:8}. The learned kernel $\hat k$ achieves stabilization in Figure \ref{fig:9} (right), but not by $t=1$, as it would with  perfect $k$
 in \eqref{eq-targetPDE}, \eqref{eq-targetPDEBC}, but only exponentially, as guaranteed for the learned $\hat k$ in Theorem \ref{thm-stabDeepONet-gf}.
 
 For this 2D problem ($f$ and $k$ are functions of $x$ and $y$), we design the branch network of the NO with convolutional neural networks (CNNs) as they have had large success in handling 2D inputs \cite{alexnet,lecunCNN}. The network consists of 70 million parameters (due to the CNNs), yet only takes around 5 minutes to train. On 900 instances, the network achieves a relative $L_2$ training error of $1.3e-3$ and a relative $L_2$ testing error of $1.8e-3$ on 100 instances.

\section{Conclusions}
\label{sec-concl}

\paragraph{What is achieved} 
PINN, DeepONet, FNO, LOCA, NOMAD---they have all been used with success to approximate solution maps of PDEs. What we introduce is a novel framework: for approximating the solution maps for {\em integral equations} \eqref{eq:kernelIntegral}, \eqref{eq-F_0}, \eqref{eq-Fgf}, or simply \eqref{eq:1.11}, for the feedback gain functions $k$ in {\em control of PDEs}. 

We provide the guarantees that (i) any desired level of accuracy of NO approximation of the backstepping gain kernel is achieved for any $\beta$ that satisfies $\|\beta\|_\infty \leq B$ for arbitrarily large given $B>0$, and (ii) the PDE is stabilized with an NO-approximated gain kernel for any $\|\beta\|_\infty \leq B$. 

These results generalize to a class of PIDEs with functional coefficients $(g,f)$ that depends on two variables, $(x,y)$, and result in kernels $k$ that are also functions of $(x,y)$.

For a given $B>0$ and any chosen positive $\epsilon < \epsilon^*(B)$, the determination of the NO approximate operator $\hat{\cal K}(\cdot)$ is done offline, once only, and such a $\hat{\cal K}(\cdot)$, which depends on $B$ and $\epsilon$, is usable ``forever,'' so to speak, for any recirculation kernel that does not violate $\|\beta\|_\infty \leq B$. 

When the entire PDE backstepping feedback law---rather than just its gain kernel---is being approximated, globality and perfect convergence are lost, but only slightly. Decay remains exponential, over infinite time, and stability is semiglobal. 

\paragraph{What is gained by making a particular controller class with theoretical guarantees the object of learning} By now it is probably clear to the reader that what we present here is a method for learning an entire {\em class} of model-based controllers, by learning the gains $\hat k=\hat{\cal K}(\beta)$, or $\hat k = {\cal Q}(g,f)$, for any plant parameters $\beta$ or $(g,f)$. What does one profit from learning a particular class of  controllers backed up by theory? Suppose that, instead of  learning  the {\em PDE backstepping} gain mapping ${\cal K}(\cdot)$, we were trying to find {\em any} gain function $k(x)$ that meets some performance objective. This goal could be formulated as a finite-time minimization of $\int_0^{t_{\rm f}} \left(\int_0^1 u^2(x,t) dx + U^2(t)\right) dt $, for a given $\beta$, over a set of gain functions $k$ for a ball of initial conditions $u_0(x) = u(x,0)$ around the origin. Not only would this be a much larger search, over  $(k,u_0)$, but such a finite-time minimization could ensure only finite-time performance, not exponential stability. 

Our achievement of global exponential stability (not ``practical''/approximate, but with an actual convergence of the state to zero) relies crucially---in each of the lemmas and theorems that we state---on the theoretical steps from the PDE backstepping toolkit (backstepping transform,  target system,  integral equation for  kernel, successive infinite-series approximation,  Lyapunov analysis). It is only by assigning the NO a service role in an otherwise model-based design that stability is assured. Stability assurance is absent from learning approaches in which the feedback law design is left to ML and a finite-time cost, as in RL for the traffic flow PDEs \cite{9568241}.

\paragraph{Future research} Of immediate  interest are the extensions of the results of this paper to 
parabolic PDEs in \cite{1369395}, as well as extensions from the approximations of controller kernels to the NO approximations of PDE backstepping observer kernels \cite{SMYSHLYAEV2005613}, with guarantees of observer convergence, and with observer-based stabilization (separation principle). 
 
\begin{figure*}
    \centering
    \includegraphics[width=\textwidth]{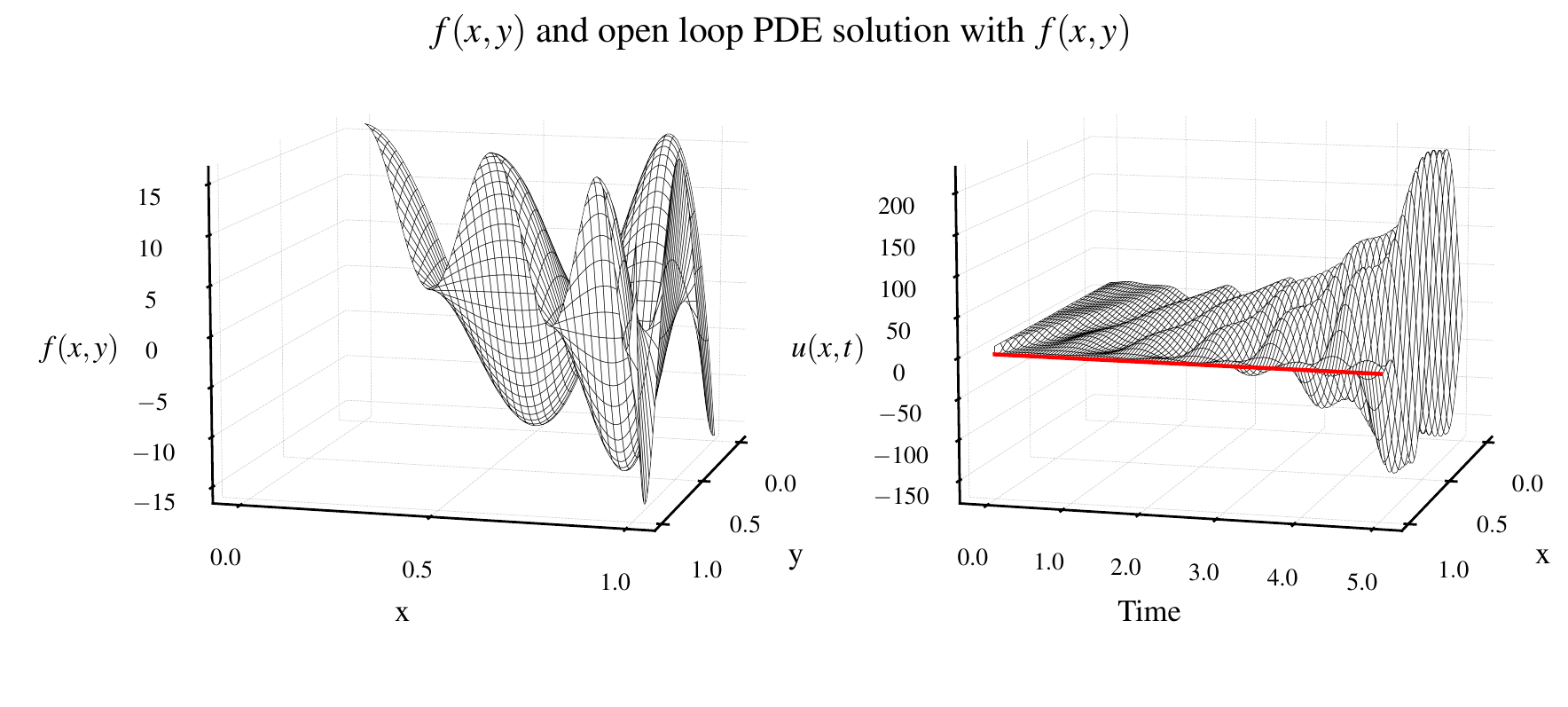}
    \caption{On the left is an example of $f(x,y)$ generated by the following form $f(x, y)$ = $\beta(x)\beta(y)$ where $\beta$ is the Chebyshev polynomial as in Fig. \ref{fig:1} with $\gamma=6$. On the right is the corresponding PDE with open loop instability.}
    \label{fig:7}
\end{figure*} 

\begin{figure*}
    \centering
    \includegraphics[width=\textwidth]{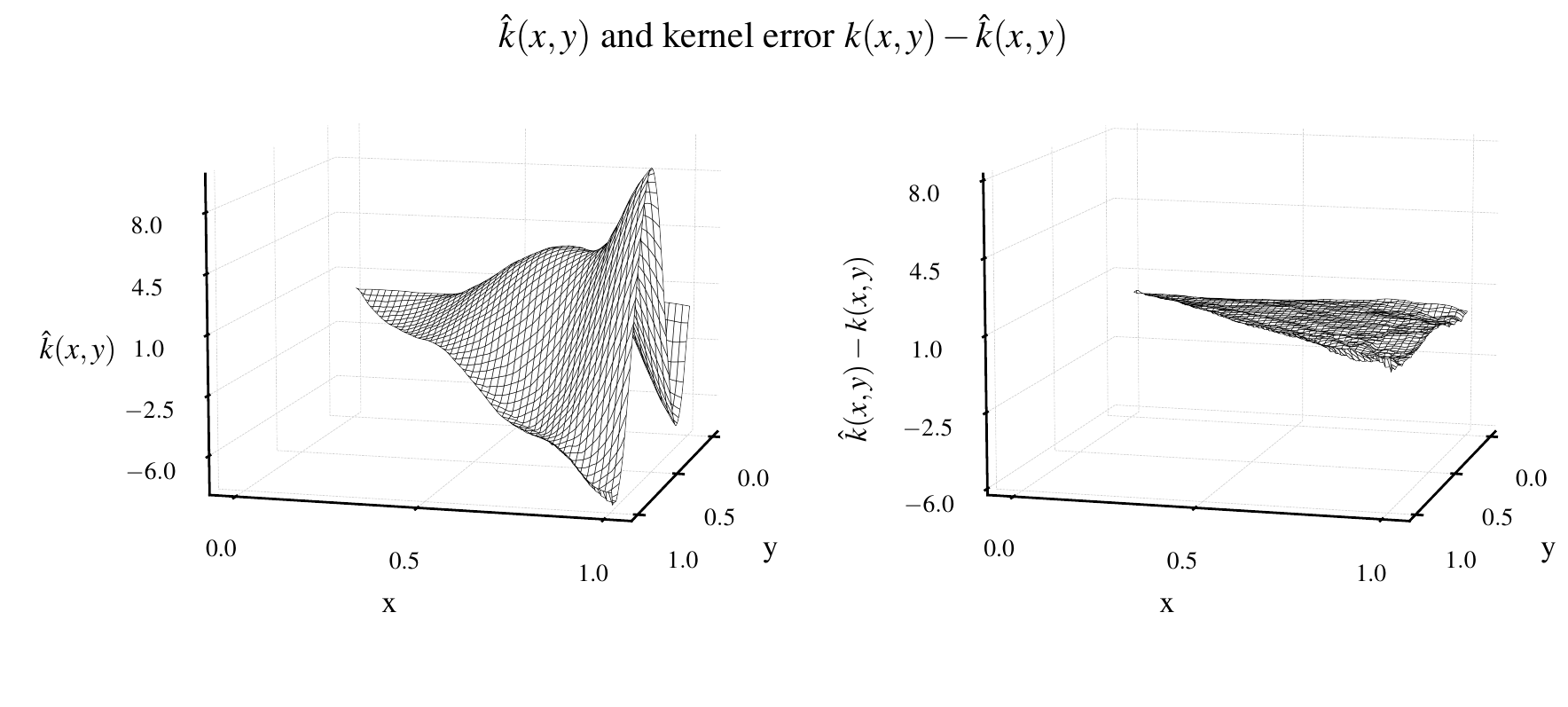}
    \caption{Examples of the learned kernel $\hat{k}(x, y)$ and the  kernel error $k(x, y) - \hat k(x, y)$ with a  peak of about $10\%$ of $k(x, y)$, around the coordinate $(0.9, 0.9)$. The kernel shown corresponds to $f(x, y)$ in Figure \ref{fig:7}.}
    \label{fig:8}
\end{figure*}

\begin{figure*}
    \centering
    \includegraphics[width=\textwidth]{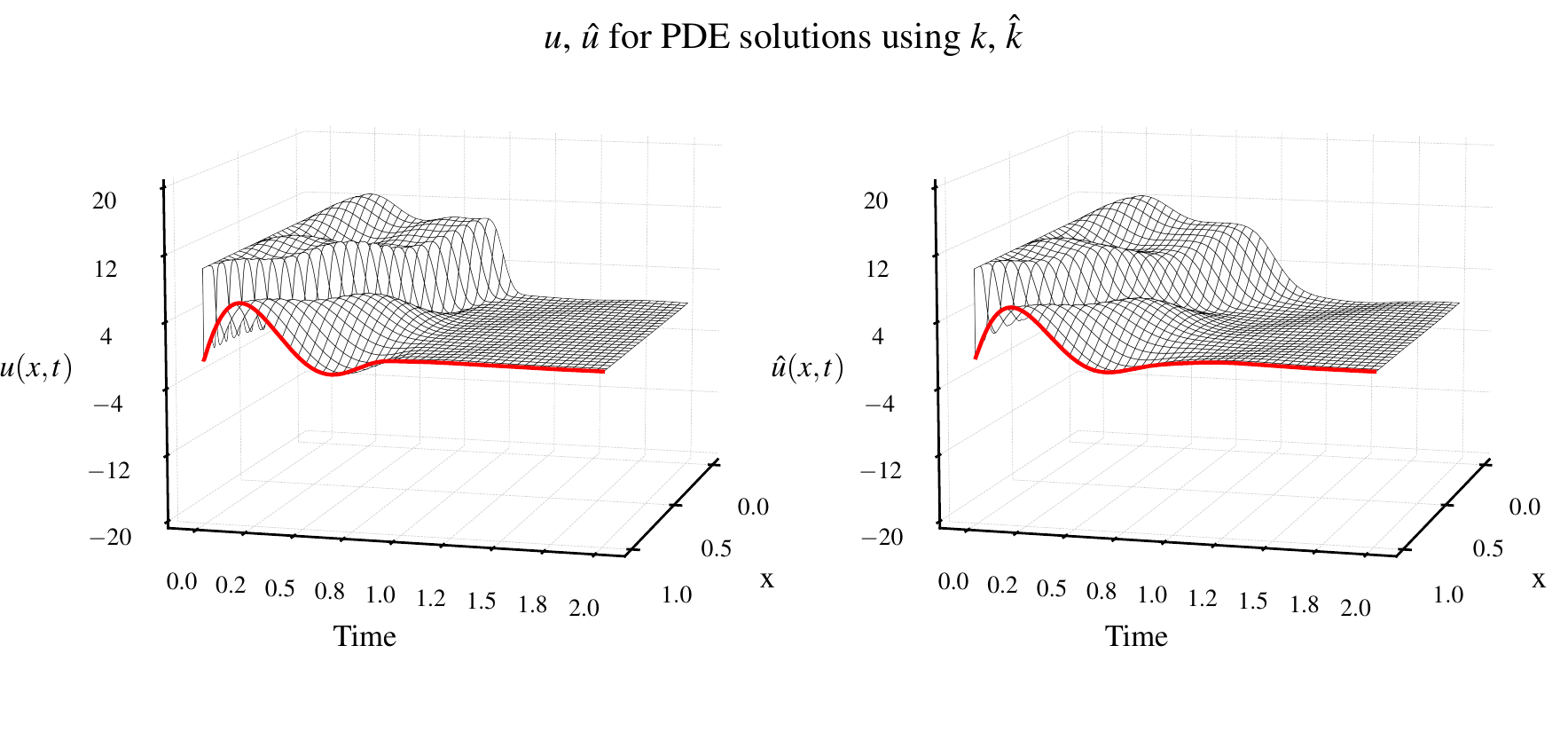}
    \caption{PDE with the analytical kernel on the
    top left and with the learned kernel on the top right. Additionally, the $L_2$ error between the systems peaks around $t=1.3$ at a value of $0.2$. The PDE shown corresponds to the $f(x, y)$ presented in Figure \ref{fig:7}. }
    \label{fig:9}
\end{figure*}

\bibliography{refs}
\bibliographystyle{abbrv}
\vskip -2\baselineskip plus -1fil

\begin{IEEEbiography}[{\includegraphics[width=1in,height=1.25in,clip,keepaspectratio]{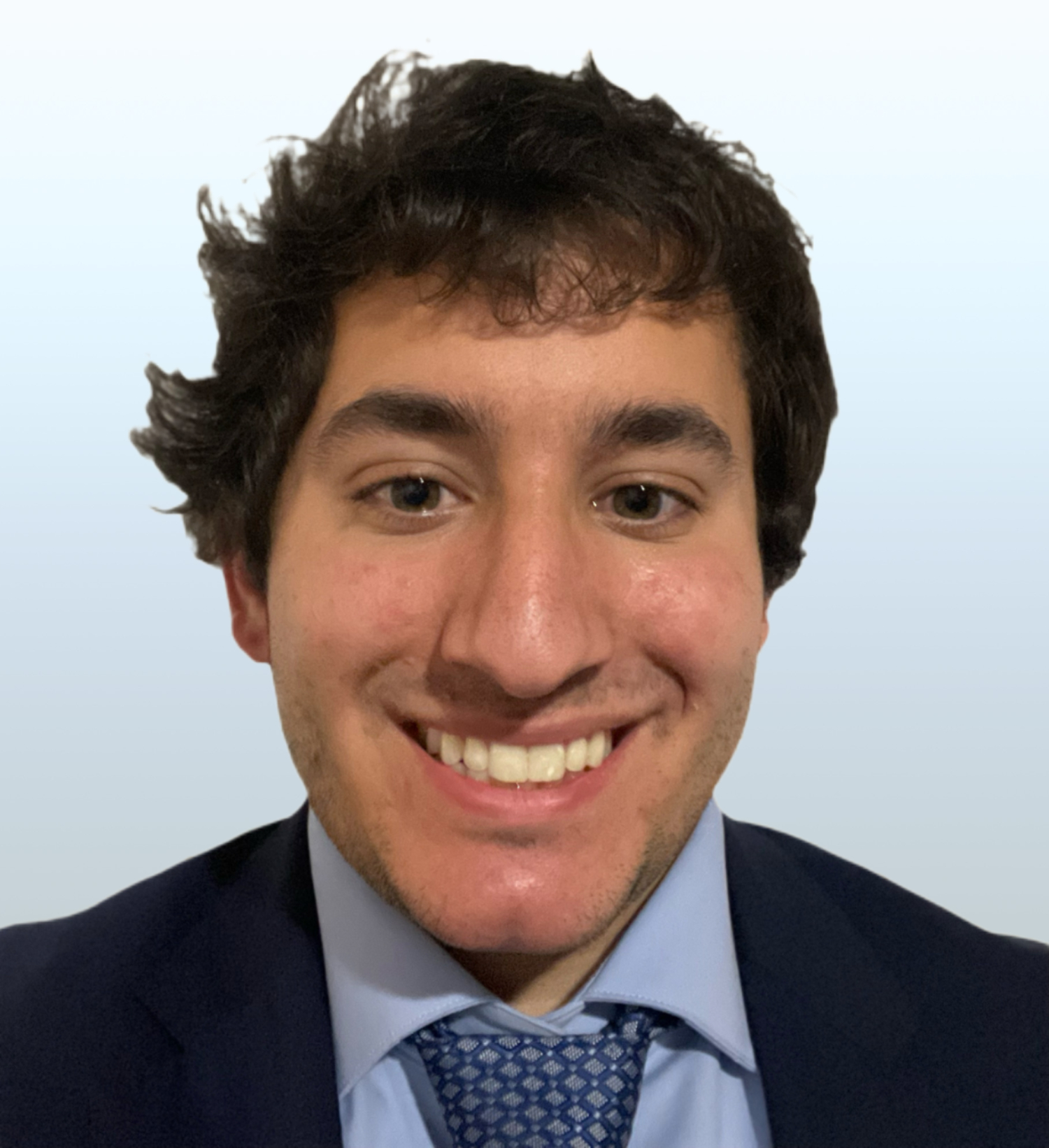}}]{Luke Bhan} received his B.S. and M.S. degrees in Computer Science and Physics from Vanderbilt University in 2022. He is currently pursuing his Ph.D. degree in Electrical and Computer Engineering
at the University of California, San Diego. His
research interests include neural operators, learning-based control, and control of partial differential equations. 
 \end{IEEEbiography}

\vskip -2\baselineskip plus -1fil

\begin{IEEEbiography}[{\includegraphics[width=1in,height=1.25in,clip,keepaspectratio]{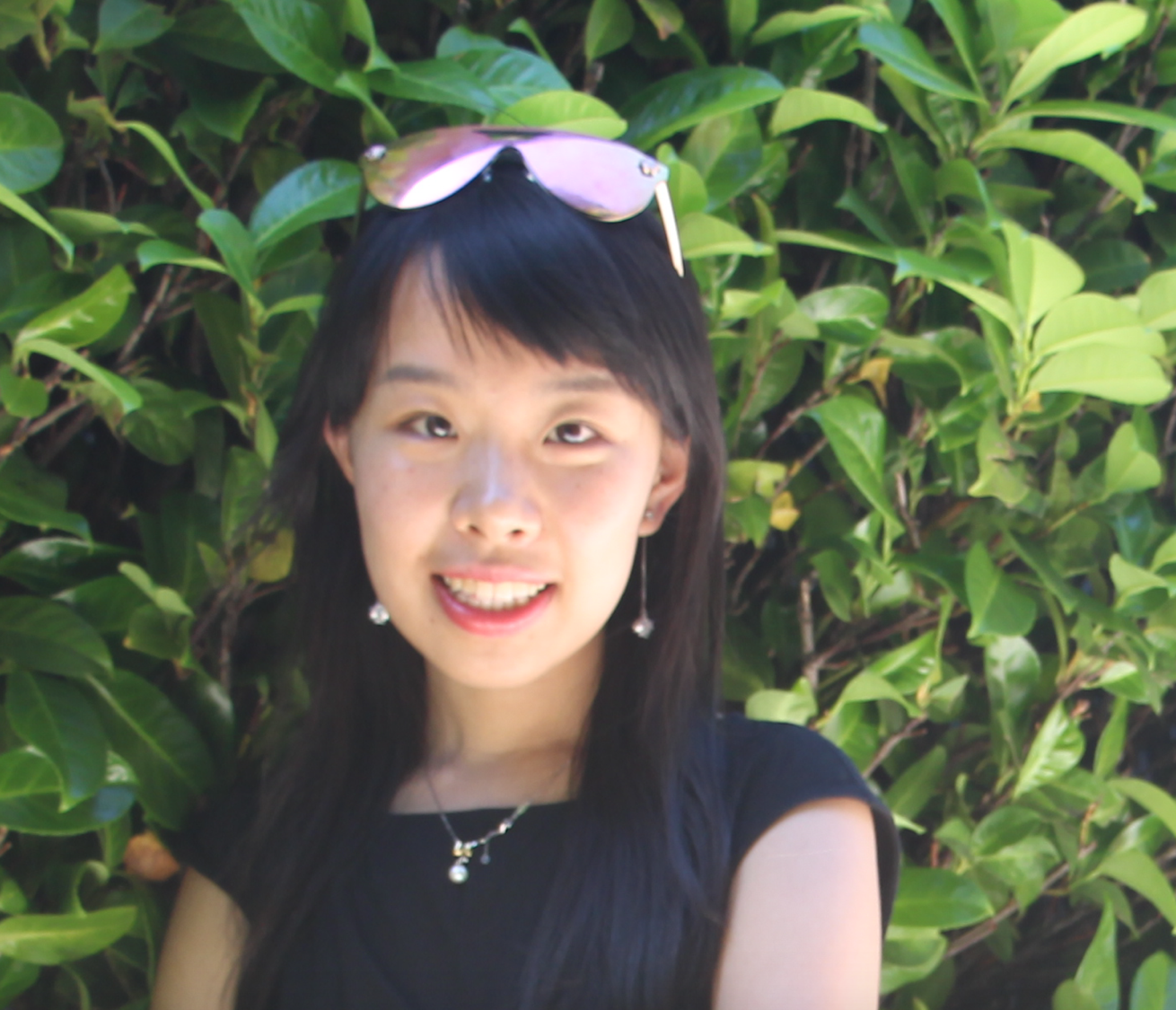}}]{Yuanyuan Shi} is an Assistant Professor of Electrical and Computer Engineering at the
University of California, San Diego. She received her Ph.D. in Electrical Engineering, masters in Electrical Engineering and Statistics, all from the University of Washington,
in 2020.
From 2020 to 2021, she was a postdoctoral scholar at 
the California Institute of Technology. Her research interests include machine learning, dynamical systems, and control, with applications to sustainable power and energy systems.
\end{IEEEbiography}

\vskip -2\baselineskip plus -1fil
\begin{IEEEbiography}[{\includegraphics[width=1in,height=1.25in,clip,keepaspectratio]{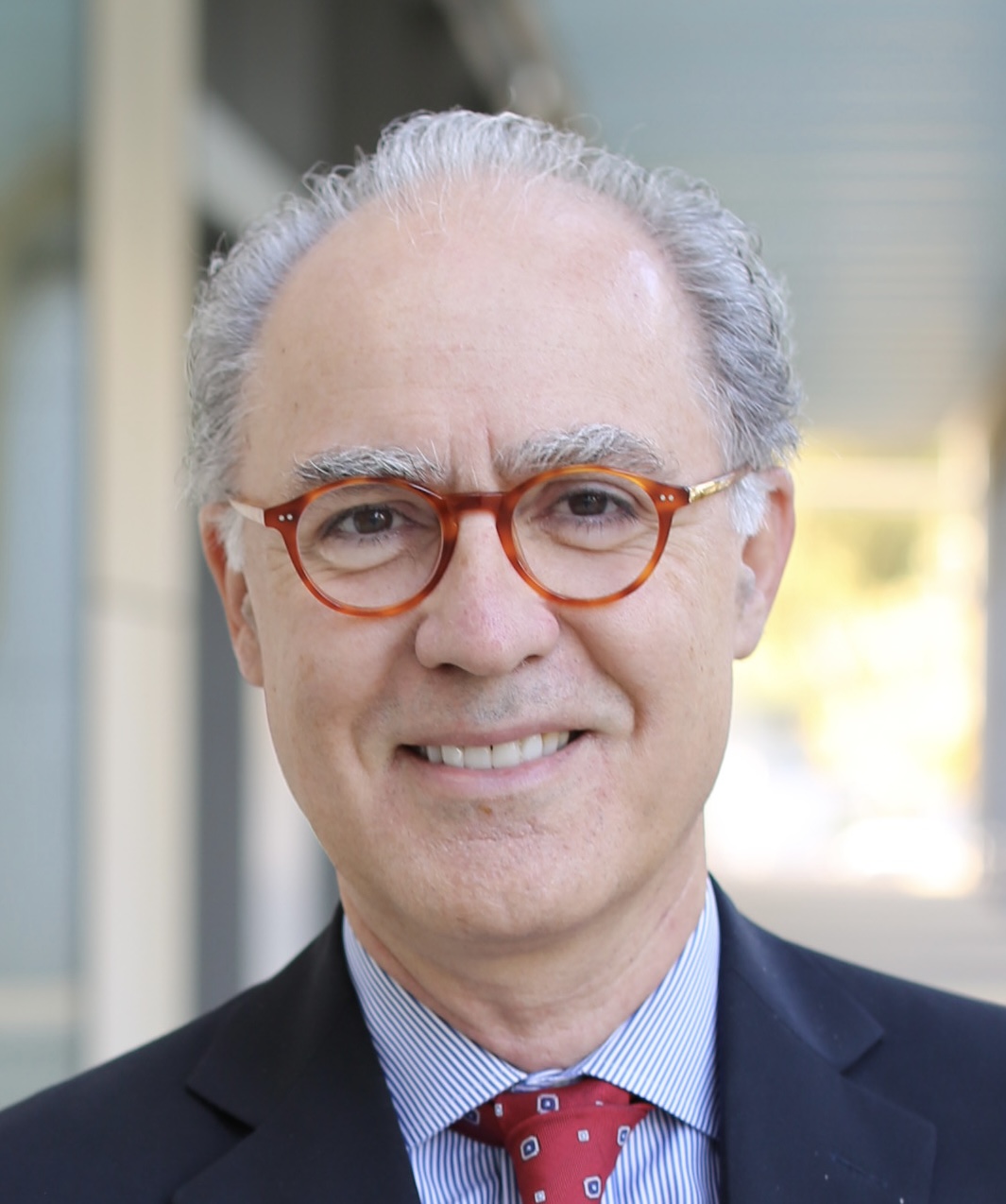}}]{Miroslav Krstic} is Distinguished Professor at UC San Diego, Alspach chair, center director, and Sr.  Assoc. Vice Chancellor for Research. He is Fellow of IEEE, IFAC, ASME, SIAM, AAAS, IET (UK), AIAA (Assoc. Fellow), and member of the Serbian Academy of Sciences and Arts, and Academy of Engineering of Serbia. He has received the Richard E. Bellman Control Heritage Award, SIAM Reid Prize, ASME Oldenburger Medal, Nyquist Lecture Prize, Paynter Outstanding Investigator Award, Ragazzini Education Award, IFAC Ruth Curtain Distributed Parameter Systems Award, IFAC Nonlinear Control Systems Award, Chestnut textbook prize, Control Systems Society Distinguished Member Award, the PECASE, NSF Career, and ONR Young Investigator awards, the Schuck (’96 and ’19) and Axelby paper prizes. He serves as Editor-in-Chief of Systems \& Control Letters, has served as SE in Automatica and IEEE TAC, and is editor of two Springer book series. Krstic has coauthored eighteen books on adaptive, nonlinear, and stochastic control, extremum seeking, control of PDE systems including turbulent flows, and control of delay systems.
\end{IEEEbiography}

\end{document}